\documentclass[11pt]{article}

\usepackage[a4paper,margin=1in]{geometry}
\usepackage{amsmath,amssymb,amsfonts,amsthm,mathtools,bm}
\usepackage{microtype}
\usepackage{enumitem}
\usepackage{booktabs}
\usepackage{array}
\usepackage[numbers,sort&compress]{natbib}
\usepackage[colorlinks=true,linkcolor=blue,citecolor=blue,urlcolor=blue]{hyperref}
\usepackage{xcolor}
\usepackage{tikz}
\usetikzlibrary{arrows.meta,positioning,calc}
\numberwithin{equation}{section}

\newtheorem{theorem}{Theorem}[section]
\newtheorem{proposition}[theorem]{Proposition}

\newtheorem{corollary}[theorem]{Corollary}
\theoremstyle{definition}

\newtheorem{example}[theorem]{Example}
\newtheorem{remark}[theorem]{Remark}

\newcommand{\C}{\mathbb{C}}
\newcommand{\K}{\mathbb{K}}

\title{\bf Multiparameter Quantum Affine Spaces and the Scalene Yang--Baxter Equation}

\author{Pramod Padmanabhan$^{a}$, Somnath Maity$^{a}$, Vladimir Korepin$^{b}$\\[1mm]
\small $^{a}$Department of Physics, School of Basic Sciences, Indian Institute of Technology,\\[-1mm]
\small Bhubaneswar, 752050, India\\
\small $^{b}$C. N. Yang Institute for Theoretical Physics, Stony Brook University,\\[-1mm]
\small New York 11794, USA\\[1mm]
\small \texttt{pramod23phys@gmail.com}, \texttt{somnathmaity126@gmail.com},\\[-1mm]
\small \texttt{vladimir.korepin@stonybrook.edu}}
\date{\today}

\begin{document}
\maketitle

\begin{abstract}
We construct representation-independent families of solutions of the
non-braided scalene Yang--Baxter equation from quadratic
noncommutative algebras. Beginning with two anticommuting generators,
we obtain a continuous deformation in terms of quantum-plane algebras
and extend the construction to an arbitrary number of generators. In
the latter case the scalene Yang--Baxter relation fixes the
multiparameter exchange matrix to an exact multiplicative form,
thereby selecting a distinguished subclass of multiparameter quantum
affine spaces. We construct finite-dimensional realizations using
singular matrices and finite Heisenberg--Weyl operators, as well as
infinite-dimensional realizations in terms of bilateral weighted shifts
and multiplicative-shift operators. 
The explicit realizations are generically weakly purely scalene:
their individual constituent operators do not satisfy the ordinary
constant non-braided Yang--Baxter equation. For selected
finite-dimensional realizations, we additionally establish strong
purity by obstructing a generic common meromorphic Yang--Baxter
parametrization, even up to consistent invertible local similarities
and nonzero scalar rescalings.
 These results provide a representation-independent algebraic framework for constructing scalene triples and a starting point for investigating their possible applications to quantum integrability.
\end{abstract}

\tableofcontents

\section{Introduction}
\label{sec:introduction}

The Yang--Baxter equation (YBE) is one of the fundamental algebraic
consistency relations underlying quantum integrability, exactly solvable lattice models, braid group representations and quantum groups
\cite{Baxter1982,KorepinBogoliubovIzergin1993,Faddeev1996}.
In its non-braided operator form it reads
\begin{equation}
R_{12}R_{13}R_{23} = R_{23}R_{13}R_{12},
\label{eq:YBE}
\end{equation}
where the same operator, or members of the same operator family, appears on the three pairs of tensor factors.  The familiar RTT/FRT construction uses solutions of the Yang--Baxter equation to define quadratic exchange algebras and, under suitable hypotheses, associated bialgebraic and quantum-group structures
\cite{FRT1989,Isaev2022}.

A broader algebraic problem is obtained by allowing three independent
operators on the three pairs of tensor factors,
\begin{equation}
A_{12}B_{13}C_{23} = C_{23}B_{13}A_{12}.
\label{eq:scalene}
\end{equation}
Following the terminology used in recent work
\cite{PadmanabhanMaityKorepin2026,konstantinou2026scalene}, we refer to
\eqref{eq:scalene} as the \emph{non-braided scalene Yang--Baxter
equation}.

This non-braided relation has two equivalent braided conventions. For \(X=A,B,C\), and \(P(x\otimes y)=y\otimes x\),
setting \(\check X=PX\) gives
\(\check A_{23}\check B_{12}\check C_{23}
=\check C_{12}\check B_{23}\check A_{12}\),
whereas setting \(\widehat X=XP\) gives
\(\widehat A_{12}\widehat B_{23}\widehat C_{12}
=\widehat C_{23}\widehat B_{12}\widehat A_{23}\).
Since \(\widehat X=P\check X P\), the transformed triples
generally differ; a fixed ordered triple need not satisfy
both braided forms. Note that when $A=B=C=R$, the two braided forms coincide taking us back to the unique braided form of standard Yang-Baxter equation \eqref{eq:YBE}. The terminology emphasizes that the three operators
$A$, $B$ and $C$ are distinct and form an ordered triple associated
with the tensor pairs $(12)$, $(13)$ and $(23)$, respectively.
Hietarinta and Viallet studied two-state five-, six- and eight-vertex
solutions of precisely this form, with particular attention to gauge,
inversion, color and spectral parametrizations
\cite{HietarintaViallet2022}.  Related multi-operator consistency
relations also arise in the theory of Yang--Baxter systems and
entwining structures
\cite{Hlavaty1997,BrzezinskiNichita2005,
BerceanuNichitaPopescu2013}.

The three-operator relation \eqref{eq:scalene}, also appears in Baxter's treatment
of the zero-field eight-vertex model
\cite[Secs.~10.4--10.5]{Baxter1982}. For generic solutions
within that ansatz, the compatible matrices admit a common
elliptic parametrization, with
\(A=R(u)\), \(B=R(u+v)\), and \(C=R(v)\), up to scalar
normalizations. Thus a scalene triple may consist of three
evaluations of a single solution of the spectral-parameter
Yang--Baxter equation. A similar structure has also appeared in the star--triangle
and associated Yang--Baxter constructions for the chiral
Potts model \cite{BaxterPerkAuYang1988}. The connections
between the spin, vertex, and interaction-round-a-face
formulations are reviewed by Perk and Au-Yang
\cite{PerkAuYang2006}.

Our interest here is more restrictive than merely finding solutions of
the scalene equation.  We seek what we shall call \emph{purely scalene} solutions. We will use two definitions for this term : {\it weak} and {\it strong}. A parametrized triple \((A,B,C)\) is weakly purely scalene if, for generic values of its parameters, none of its constituents individually satisfies the ordinary constant non-braided Yang–Baxter equation:
$$ R_{12}R_{13}R_{23}\neq R_{23}R_{13}R_{12}, \qquad R\in\{A,B,C\}. $$
This definition refers to the chosen realization.

A triple is strongly purely scalene if it is weakly
purely scalene and, up to nonzero scalar factors and consistent
invertible local changes of basis at the three sites, its
parametrized family cannot generically be written as
\[
A=R(u,v;\mu),\qquad
B=R(u,w;\mu),\qquad
C=R(v,w;\mu),
\]
where $R$ is a single meromorphic Yang--Baxter family on a connected
spectral domain, $u,v,w$ vary independently, and the shared moduli
$\mu$ are fixed. Embeddings at specially constrained spectral
values do not by themselves violate this definition. The scalene triples constructed in this work are representation-independent and hence algebraic. For explicit realizations the scalene triples constructed in this work are purely scalene by the weak definition. On the other hand only some of them are purely scalene in the strong sense. The latter will be explicitly identified.


The search for such solutions is motivated in part by the possibility
that the scalene equation may provide algebraic structures relevant to
integrability beyond the usual single-$R$-matrix framework.  In our
recent work we found examples in which individual members of a scalene
triple fail the ordinary YBE, while the mixed relation nevertheless
leads to cross-commuting transfer matrices and nontrivial hidden
symmetries \cite{PadmanabhanMaityKorepin2026}.  Cross-commutativity
alone, however, does not imply the existence of a conventional
self-commuting transfer-matrix family.  We therefore distinguish
throughout between the construction of purely scalene solutions and the
further problem of extracting integrable lattice models from them.

The approach of the present paper is deliberately algebraic.  We seek
representation-independent solutions of \eqref{eq:scalene} generated
by quadratic noncommutative algebras.  This continues a line of work in
which algebraic ans\"atze, and in particular Clifford-algebra
constructions, were used to obtain solutions of the Yang--Baxter,
tetrahedron and higher simplex equations
\cite{PadmanabhanKorepin2024Clifford,
PadmanabhanSinghKorepin2025CliffordSolver}.
Related algebraic methods were also used to reproduce and organize the
invertible constant $4\times4$ solutions in Hietarinta's classification
\cite{MaitySinghPadmanabhanKorepin2024}.  Here the objective is instead
to identify an algebraic mechanism intrinsic to the scalene relation.

We begin with three independent associative algebras, each containing
two anticommuting generators.  No relations on the squares of these
generators are required, so the construction is representation
independent and is not restricted to Clifford algebras.  Replacing
anticommutation by $q$-commutation leads naturally to quantum-plane,
or Manin-plane, algebras
\cite{Manin1988,Manin1987Koszul,Manin2018QuantumGroups}.  The scalene equation fixes the relative
orientation of the three deformation parameters.

The construction extends to an arbitrary number $N$ of generators.
In this case the three local algebras are multiparameter quantum affine
spaces.  The scalene Yang--Baxter equation imposes a strong
compatibility condition on their exchange parameters and selects a
distinguished exact subclass of the general multiparameter quantum
affine spaces
\cite{GoodearlLetzter2008,MukherjeeBera2024}.  We then construct
finite-dimensional representations, including singular realizations
and finite Heisenberg--Weyl representations at roots of unity, as well
as infinite-dimensional realizations in terms of bilateral weighted
shifts and multiplicative-shift operators.  The explicit
representations considered in the paper are generically purely
scalene in the sense defined earlier.

The paper is organized as follows.
Section~\ref{sec:twogenerator} develops the representation-independent
two-generator anticommuting construction and gives its Pauli
realization.  Section~\ref{sec:qdeform} introduces the quantum-plane
deformation.  Section~\ref{sec:Ngenerator} extends the construction to
$N$ generators and classifies the allowed multiparameter exchange
relations.  Section~\ref{sec:finite} studies finite-dimensional
representations, including singular and finite Weyl realizations.
Section~\ref{sec:infinite} develops the infinite-dimensional
weighted-shift and function-space representations.
Section~\ref{sec:conclusions} summarizes the purely scalene character
of these solutions, their cross-commutativity structure, and directions
for constructing integrable models.  Additional calculations are
collected in the appendices.

\section{Two-generator algebras : Anticommuting case}
\label{sec:twogenerator}
Consider three associative and noncommutative algebras denoted $$\mathcal S_\alpha~~;~~\alpha\in\{I,II,III \},$$  over a field $\C$ of complex numbers\footnote{We can replace this with any field $\mathbb{F}$, of characteristic different from two.}. Each of these algebras is generated by two elements. They can be collectively written as
\begin{equation}
\left\{X_\alpha^{(1)},X_\alpha^{(2)}\right\}\in\mathcal S_\alpha,\qquad \alpha\in\{I,II,III\}.
\end{equation}
Within each algebra they satisfy the relations
\begin{equation}
X_\alpha^{(1)}X_\alpha^{(2)} = -X_\alpha^{(2)}X_\alpha^{(1)},\qquad \alpha\in\{I,II,III\}.
\label{eq:anti}
\end{equation}
There are no relations imposed between the generators from two different algebras $\mathcal{S}_\alpha$'s.

When considering these operators acting on different sites of the one-dimensional lattice that we will work on, the notation for these generators has to be augmented with site indices as well. This will be denoted $$ X_{j;\alpha}^{(\mu)}. $$
Here the indices $\alpha$ and $\mu$ specify the algebra and the generator number, respectively and the index $j$ keeps track of the lattice sites. The generators of the same or different algebras acting on different sites commute with each other. This is consistent with the tensor product nature of the full Hilbert space. 

The range of values the indices $\alpha$ and $j$ assume will remain the same throughout this paper. On the other hand, the range of values index $\mu$ takes will be allowed to increase and will eventually be $N$. That is, we allow the number of generators to increase while keeping the number of associative algebras and the number of lattice sites to be the same.

With the above notation, define the following triple of operators
\begin{subequations}
\label{eq:twogenerator}
\begin{align}
A_{12}&=a~X_{1;I}^{(1)}X_{2;II}^{(1)}+b~X_{1;I}^{(2)}X_{2;II}^{(2)},\label{eq:alpha2}\\
B_{13}&=c~X_{1;I}^{(1)}X_{3;III}^{(1)}+d~X_{1;I}^{(2)}X_{3;III}^{(2)},\label{eq:beta2}\\
C_{23}&=e~X_{2;II}^{(1)}X_{3;III}^{(1)}+f~X_{2;II}^{(2)}X_{3;III}^{(2)},\label{eq:gamma2}
\end{align}
\end{subequations}
where $a,b,c,d,e,f\in\C$ are arbitrary scalars. They can also be viewed as arbitrary spectral parameter-dependent functions. These ans\"atze are a generalization of the Clifford solutions of the standard Yang-Baxter equation introduced in \cite{PadmanabhanKorepin2024Clifford}.

\begin{theorem}[Two-generator anticommuting construction]
\label{thm:anticomm}
The operators \eqref{eq:twogenerator} obey the scalene Yang--Baxter relation
\begin{equation}
A_{12}B_{13}C_{23}
=
C_{23}B_{13}A_{12}
\label{eq:SYBabc}
\end{equation}
for arbitrary coefficients $a,b,c,d,e,f$.
\end{theorem}

\begin{proof}
This can be proved in two ways. One proof follows the algebraic
strategy used in Ref.~\cite{PadmanabhanKorepin2024Clifford}, while
the other follows from a direct expansion of both sides of the
scalene relation. We give the latter explicitly in Appendix \ref{app:twogenerator}. 

A short summary of why the latter works is as follows. As operators on different tensor factors commute, every monomial can be grouped site by site.  Expanding the coefficient of a fixed product of $a,b,c,d,e,f$ amounts to reversing local orders at a subset of the three sites.  Each mixed term contains an even number of local reversals, so the minus signs generated by \eqref{eq:anti} cancel pairwise.  The two unmixed terms are unchanged.  Therefore the eight coefficient sectors agree separately. 
\end{proof}

\begin{remark}
    An important feature of Theorem~\ref{thm:anticomm} is that it holds for arbitrary values of the parameters $a,\cdots, f$. If they are taken to be spectral parameter-dependent functions, then they can be arbitrarily chosen from a set of free functions. It should be noted that an arbitrary spectral dependence of the coefficients preserves the mixed scalene relation, it does not automatically supply a common \(R\)-matrix family.
\end{remark}

\begin{remark}
Theorem~\ref{thm:anticomm} is often conveniently realized with Clifford or Pauli generators, in the spirit of the representation-independent Clifford constructions of Refs.~\cite{PadmanabhanKorepin2024Clifford,PadmanabhanSinghKorepin2025CliffordSolver}, but the theorem itself is more general.  A Clifford algebra would additionally impose quadratic relations on the squares.  Here the universal local algebra generated by $X^{(1)}_{j;\alpha}, X^{(2)}_{j;\alpha}$ is simply
\begin{equation}
\C\langle X^{(1)}_{j;\alpha},X^{(2)}_{j;\alpha}\rangle/(X^{(1)}_{j;\alpha} X^{(2)}_{j;\alpha} + X^{(2)}_{j;\alpha} X^{(1)}_{j;\alpha})~~;~~\alpha\in\{I,II,III\}
\end{equation}
on a site $j$.  The Clifford language therefore describes a distinguished class of representations, not the universal algebraic input.
\end{remark}

\begin{example}[Pauli realization]
As a concrete example consider a two-dimensional representation for the three algebras realized by the Pauli matrices
\begin{equation}
\left(X^{(1)}_{j;I},X^{(2)}_{j;I}\right)\equiv(\sigma^x_j,\sigma^y_j),\qquad \left(X^{(1)}_{j;II},X^{(2)}_{j;II}\right)\equiv(\sigma^y_j,\sigma^z_j),\qquad \left(X^{(1)}_{j;III},X^{(2)}_{j;III}\right)\equiv(\sigma^z_j,\sigma^x_j),
\end{equation}
where $\sigma^x,\sigma^y,\sigma^z$ are the spin $\frac{1}{2}$ Pauli matrices.  Then
\begin{equation}
A_{12}=a~\sigma^x_1\sigma^y_2+b~\sigma^y_1\sigma^z_2~;~
B_{13}=c~\sigma^x_1\sigma^z_3+d~\sigma^y_1\sigma^x_3~;~
C_{23}=e~\sigma^y_2\sigma^z_3+f~\sigma^z_2\sigma^x_3,
\end{equation}
and Theorem~\ref{thm:anticomm} gives \eqref{eq:SYBabc} identically.  As concrete $4\times4$ matrices these are,
\begin{eqnarray}
& A =
\begin{pmatrix}
0 & 0 & -\mathrm{i}b & -\mathrm{i}a \\
0 & 0 & \mathrm{i}a & \mathrm{i}b \\
\mathrm{i}b & -\mathrm{i}a & 0 & 0 \\
\mathrm{i}a & -\mathrm{i}b & 0 & 0
\end{pmatrix}~;~B =
\begin{pmatrix}
0 & 0 & c & -\mathrm{i}d \\
0 & 0 & -\mathrm{i}d & -c \\
c & \mathrm{i}d & 0 & 0 \\
\mathrm{i}d & -c & 0 & 0
\end{pmatrix}, &
\nonumber \\[1em]
& C =
\begin{pmatrix}
0 & f & -\mathrm{i}e & 0 \\
f & 0 & 0 & \mathrm{i}e \\
\mathrm{i}e & 0 & 0 & -f \\
0 & -\mathrm{i}e & -f & 0
\end{pmatrix}. &
\end{eqnarray}
Furthermore, none of $A$, $B$ or $C$ satisfy the standard non-braided YBE for non-zero choices of the parameters $a,\cdots, f$. That is $$ R_{12}R_{13}R_{23} \neq R_{23}R_{13}R_{12}~;~R\in\{A,B,C\}.$$
Thus this simple example satisfies the weak definition for a purely scalene triple.

On the other hand it fails the strong definition.
Let $a,c,e\neq 0$. Then consider the parameters
\[
r_A=\frac ba,\qquad r_B=\frac dc,\qquad r_C=\frac fe,
\]
and the $R$-matrix
\[
R_\kappa(z)=X\otimes X+(\kappa+z)Y\otimes Y,
\]
with 
\[
\kappa=r_A+r_C-r_B,\qquad
s=r_B-r_C,\qquad t=r_B-r_A.
\]
Then for the given scalene triple there exist compatible independent rotations
of the three local spaces that transform the triple into
\[
A'=aR_\kappa(s),\qquad
B'=cR_\kappa(s+t),\qquad
C'=eR_\kappa(t).
\]
The common family $R_\kappa(z)$
satisfies the mixed Yang--Baxter relation for arbitrary
independent coefficient pairs \cite{PadmanabhanKorepin2024Clifford}.
Thus this example admits a common spectral embedding after
local basis changes.

\end{example}

The structure of \eqref{eq:twogenerator} is naturally visualized on a triangle (See Figure \ref{fig:triangle_twogenerator}).
\begin{figure}[t]
\centering
\begin{tikzpicture}[
    scale=1.0,
    vertex/.style={circle, draw, thick, minimum size=8mm, inner sep=0pt},
    edgelabel/.style={fill=white, inner sep=1.5pt, font=\small, align=center},
    gen1/.style={text=blue!70!black, font=\small},
    gen2/.style={text=red!75!black, font=\small},
    smallnote/.style={font=\small}
]

\node[vertex] (1) at (90:2.2) {$1$};
\node[vertex] (2) at (210:2.2) {$2$};
\node[vertex] (3) at (330:2.2) {$3$};

\draw[thick] (1) -- (2);
\draw[thick] (1) -- (3);
\draw[thick] (2) -- (3);

\node[edgelabel] at ($(1)!0.5!(2)+(-0.55,0.15)$)
{$A_{12}=a\,X_1^{(1)}X_2^{(1)}$\\[-1mm]
$\hphantom{R_{12}=}{}+\,b\,X_1^{(2)}X_2^{(2)}$};

\node[edgelabel] at ($(1)!0.5!(3)+(0.55,0.15)$)
{$B_{13}=c\,X_1^{(1)}X_3^{(1)}$\\[-1mm]
$\hphantom{R_{13}=}{}+\,d\,X_1^{(2)}X_3^{(2)}$};

\node[edgelabel] at ($(2)!0.5!(3)+(0,-0.45)$)
{$C_{23}=e\,X_2^{(1)}X_3^{(1)}$\\[-1mm]
$\hphantom{R_{23}=}{}+\,f\,X_2^{(2)}X_3^{(2)}$};

\node[gen1] at ($(1)+(-0.85,0.45)$) {$X_1^{(1)}$};
\node[gen2] at ($(1)+(0.85,0.45)$) {$X_1^{(2)}$};

\node[gen1] at ($(2)+(-0.95,-0.35)$) {$X_2^{(1)}$};
\node[gen2] at ($(2)+(0.00,-0.80)$) {$X_2^{(2)}$};

\node[gen1] at ($(3)+(0.95,-0.35)$) {$X_3^{(1)}$};
\node[gen2] at ($(3)+(0.00,-0.80)$) {$X_3^{(2)}$};

\node[smallnote] at (0,-3.15)
{$A_{12}B_{13}C_{23}\ \longleftrightarrow\ C_{23}B_{13}A_{12}$};

\end{tikzpicture}
\caption{Triangular organization of the anticommuting two-generator construction.
The blue and red labels indicate the two matched generator families
$\{X_1^{(1)},X_2^{(1)},X_3^{(1)}\}$ and
$\{X_1^{(2)},X_2^{(2)},X_3^{(2)}\}$.
Each scalene operator is an arbitrary linear combination of the two matched-generator contributions.
For mixed monomials, reversing the order
$A_{12}B_{13}C_{23}\mapsto C_{23}B_{13}A_{12}$
reverses the local generator order at an even number of vertices, so the anticommutation signs cancel.}
\label{fig:triangle_twogenerator}
\end{figure}
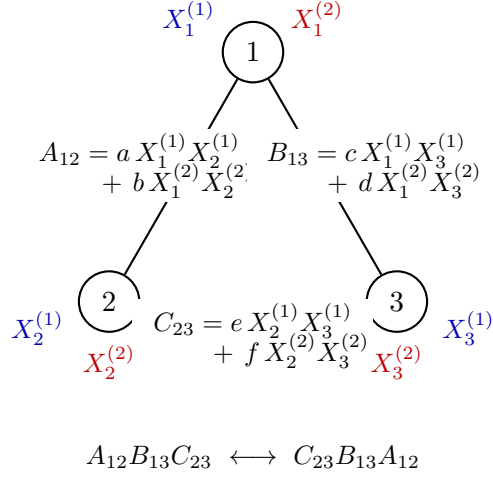
There are two matched generator families, one built from the first local generator at every vertex and the other from the second.  Each edge is an arbitrary linear combination of the two matched generator contributions. Reversing the order of the three scalene operators reverses the order of local generators at pairs of vertices.  The anticommutation signs therefore appear with even parity. 

The diagrammatics suggest that the anticommuting structure can be replaced by a $q$-deformation. And indeed in the next section we will show that this parity cancellation is nothing but a $q=-1$ specialization of a more general multiplicative phase cancellation.

\section{Two-generator algebras : $q$-deformed case}
\label{sec:qdeform}
We now relax the anticommutation relations \eqref{eq:anti} while keeping the  two-generator structure of the three associative algebras 
$\mathcal S_\alpha$, $\alpha\in\{I,II,III\}$. We replace \eqref{eq:anti} by
\begin{equation}
X_{j;\alpha}^{(1)}X_{j;\alpha}^{(2)}
=
q_\alpha~
X_{j;\alpha}^{(2)}X_{j;\alpha}^{(1)},
\qquad 
q_\alpha\in\C^\times,
\qquad
\alpha\in\{I,II,III\}.
\label{eq:qplane_local}
\end{equation}
As in the previous section, no relations are imposed between generators
belonging to different algebras.  When they act on distinct lattice sites,
they commute as a consequence of the tensor product structure.
We retain exactly the same ans\"atze for the three operators as in
\eqref{eq:twogenerator}.
The coefficients $a,b,c,d,e,f\in\C$ are again arbitrary, and may in
particular be arbitrary functions of spectral parameters. The scalene
relation now determines how the three deformation parameters
$q_I,q_{II},q_{III}$ must be related.

\begin{theorem}[Two-generator $q$-deformed construction]
\label{thm:q2}
Assume that the ordered monomials in each of the three algebras
$\mathcal S_\alpha$ are linearly independent. Then the operators
\eqref{eq:twogenerator} satisfy
\begin{equation}
A_{12}B_{13}C_{23}
=
C_{23}B_{13}A_{12}
\label{eq:SYBq}
\end{equation}
for arbitrary $a,b,c,d,e,f$ if and only if
\begin{equation}
q_I=q_{III}=q,
\qquad
q_{II}=q^{-1},
\qquad
q\in\C^\times.
\label{eq:qorientation}
\end{equation}
\end{theorem}

\begin{proof}
For notational convenience, only within this proof, define
\begin{equation}
(x_I,y_I)
\equiv
\left(X_{1;I}^{(1)},X_{1;I}^{(2)}\right)~;~
(x_{II},y_{II})
\equiv
\left(X_{2;II}^{(1)},X_{2;II}^{(2)}\right)~;~
(x_{III},y_{III})
\equiv
\left(X_{3;III}^{(1)},X_{3;III}^{(2)}\right).
\end{equation}
The local relations \eqref{eq:qplane_local} then become
\begin{equation}
x_\alpha y_\alpha
=
q_\alpha y_\alpha x_\alpha,
\qquad
\alpha\in\{I,II,III\}.
\label{eq:qshortrelation}
\end{equation}
We choose the canonical ordering in which $x_\alpha$ precedes $y_\alpha$
at each site.  Since operators acting on distinct sites commute, every
monomial occurring in the cubic products can first be grouped site by site
and then brought to this canonical local ordering.

Expanding the left hand side of \eqref{eq:SYBq} gives
\begin{align}
A_{12}B_{13}C_{23}
={}&
ace\,x_I^2x_{II}^2x_{III}^2
+
acf\,x_I^2x_{II}y_{II}x_{III}y_{III}
\notag\\
&+
\frac{ade}{q_{III}}\,
x_Iy_Ix_{II}^2x_{III}y_{III}
+
adf\,
x_Iy_Ix_{II}y_{II}y_{III}^2
\notag\\
&+
\frac{bce}{q_Iq_{II}}\,
x_Iy_Ix_{II}y_{II}x_{III}^2
+
\frac{bcf}{q_I}\,
x_Iy_Iy_{II}^2x_{III}y_{III}
\notag\\
&+
\frac{bde}{q_{II}q_{III}}\,
y_I^2x_{II}y_{II}x_{III}y_{III}
+
bdf\,
y_I^2y_{II}^2y_{III}^2 .
\label{eq:qLHS}
\end{align}
On the other hand, reversing the order of the three scalene operators gives
\begin{align}
C_{23}B_{13}A_{12}
={}&
ace\,x_I^2x_{II}^2x_{III}^2
+
\frac{acf}{q_{II}q_{III}}\,
x_I^2x_{II}y_{II}x_{III}y_{III}
\notag\\
&+
\frac{ade}{q_I}\,
x_Iy_Ix_{II}^2x_{III}y_{III}
+
\frac{adf}{q_Iq_{II}}\,
x_Iy_Ix_{II}y_{II}y_{III}^2
\notag\\
&+
bce\,
x_Iy_Ix_{II}y_{II}x_{III}^2
+
\frac{bcf}{q_{III}}\,
x_Iy_Iy_{II}^2x_{III}y_{III}
\notag\\
&+
bde\,
y_I^2x_{II}y_{II}x_{III}y_{III}
+
bdf\,
y_I^2y_{II}^2y_{III}^2 .
\label{eq:qRHS}
\end{align}

Since the six coefficients $a,b,c,d,e,f$ are arbitrary and the ordered
monomials are assumed to be linearly independent, the corresponding
coefficient sectors on the two sides must agree separately. Comparing terms on both sides of the scalene relations gives 
\begin{equation}
q_{II}q_{III}=1~;~q_Iq_{II}=1~;~q_I=q_{III}
\label{eq:qcond1}
\end{equation}
which implies that
\begin{equation}
q_I=q_{III}=q,
\qquad
q_{II}=q^{-1},
\end{equation}
proving \eqref{eq:qorientation}.
Conversely, substitution of these relations into
\eqref{eq:qLHS} and \eqref{eq:qRHS} makes all eight coefficient sectors
identical, proving \eqref{eq:SYBq}.
\end{proof}

The three local algebras selected by Theorem~\ref{thm:q2} can therefore
be written as
\begin{subequations}
\label{eq:threeqplanes}
\begin{align}
\mathcal S_I(q)
&=
\C\langle X_I^{(1)},X_I^{(2)}\rangle
\Big/
\left(
X_I^{(1)}X_I^{(2)}
-
qX_I^{(2)}X_I^{(1)}
\right),
\\
\mathcal S_{II}(q^{-1})
&=
\C\langle X_{II}^{(1)},X_{II}^{(2)}\rangle
\Big/
\left(
X_{II}^{(1)}X_{II}^{(2)}
-
q^{-1}X_{II}^{(2)}X_{II}^{(1)}
\right),
\\
\mathcal S_{III}(q)
&=
\C\langle X_{III}^{(1)},X_{III}^{(2)}\rangle
\Big/
\left(
X_{III}^{(1)}X_{III}^{(2)}
-
qX_{III}^{(2)}X_{III}^{(1)}
\right).
\end{align}
\end{subequations}
For generic $q$, an algebra generated by two elements $x,y$ satisfying
\begin{equation}
xy=qyx
\end{equation}
is the standard quantum plane or the Manin plane \cite{Manin1987Koszul,Manin1988,Manin2018QuantumGroups}. Thus Theorem~\ref{thm:q2} replaces the
three anticommuting local algebras of Section~\ref{sec:twogenerator} by
three quantum-plane algebras, with the deformation parameter at the
middle algebra inverted relative to the other two.

The anticommuting construction is recovered at $q=-1$, since $q^{-1}=q$ in this case. At $q=1$, the two generators commute within each of the three local algebras. For generic $q$, however, the three local algebras are genuinely noncommutative quantum planes \cite{Manin1988,Schmudgen2002}.

\begin{remark}[Orientation of the deformation parameter]
\label{rem:qorientation}
The pattern $q_I=q$, $q_{II}=q^{-1}$ and $q_{III}=q$
has a natural oriented interpretation. Indeed, the middle relation
\begin{equation}
X_{II}^{(1)}X_{II}^{(2)}
=
q^{-1}X_{II}^{(2)}X_{II}^{(1)}
\end{equation}
can equivalently be written as
\begin{equation}
X_{II}^{(2)}X_{II}^{(1)}
=
qX_{II}^{(1)}X_{II}^{(2)}.
\end{equation}
Thus the same deformation parameter $q$ may be associated with all three vertices provided the ordering of the two generators at the middle vertex is reversed. The inversion of $q$ is therefore not an additional assumption; it is forced by the scalene Yang--Baxter relation. In the next section this oriented phase cancellation will emerge as the two-generator specialization of the general $N$-generator construction.
\end{remark}

\section{$N$-generator algebras : multiparameter deformations}
\label{sec:Ngenerator}
We continue working with the three associative and noncommutative algebras $\mathcal S_\alpha$, with $\alpha\in\{I,II,III\}$,
but this time with $N$ generators. They are defined by the sets
\begin{equation}
\left\{
X_{j;\alpha}^{(\mu)}\Big|\mu\in\{1,2,\cdots,N\}
\right\}~;~
\alpha\in\{I,II,III\}.
\label{eq:Ngenerators}
\end{equation}
As before, the index $\alpha$ specifies the algebra, the index $\mu$
specifies the generator number, and $j$ denotes the lattice site on which the generator acts. Within each algebra we impose the multiparameter quadratic relations
\begin{equation}
X_{j;\alpha}^{(\mu)}X_{j;\alpha}^{(\nu)}
=
q_{\alpha}^{\mu\nu}
X_{j;\alpha}^{(\nu)}X_{j;\alpha}^{(\mu)},
\qquad
\alpha\in\{I,II,III\},
\qquad
\mu,\nu\in\{1,\ldots,N\}.
\label{eq:qimunu}
\end{equation}
The following constraints on the non-zero deformation parameters automatically follow from the defining relations of the generators,
\begin{equation}
q_{\alpha}^{\mu\mu}=1,
\qquad
q_{\alpha}^{\mu\nu}q_{\alpha}^{\nu\mu}=1.
\label{eq:antisymq}
\end{equation}
No relations are imposed between generators belonging to different
algebras. As in Section~\ref{sec:twogenerator}, generators acting on
different lattice sites commute with one another as a consequence of the tensor product structure.

The algebras defined by \eqref{eq:qimunu} are multiparameter quantum
affine spaces \cite{Manin1989Multiparametric}, generalizing the two-generator quantum or Manin plane of the previous section.  We now determine the restrictions on the parameters $q_{\alpha}^{\mu\nu}$ that are imposed by the scalene Yang--Baxter relation.

The ansatz for the scalene triple is a $N$-generator generalization of \eqref{eq:twogenerator}:
\begin{equation}
\label{eq:Nedge}
A_{12} = \sum_{\mu=1}^{N} a_{\mu}\, X_{1;I}^{(\mu)} X_{2;II}^{(\mu)}~;~
B_{13} = \sum_{\nu=1}^{N} b_{\nu}\, X_{1;I}^{(\nu)}X_{3;III}^{(\nu)}~;~
C_{23}=\sum_{\rho=1}^{N}c_{\rho}\,X_{2;II}^{(\rho)}X_{3;III}^{(\rho)}.
\end{equation}
Here the coefficient families $\{a_\mu\}$, $\{b_\mu\}$ and $\{c_\mu\}$ 
are arbitrary complex scalars. As in the two-generator construction,
they may equally well be taken to be arbitrary spectral parameter-dependent functions. Before we find the conditions on these multiple parameters imposed by the scalene relation, we need a notion of linear independence of operators in the associative algebras. This is given by the Poincar\'e--Birkhoff--Witt (PBW) property. For each algebra $\mathcal S_\alpha$, choose the ordering
\begin{equation}
X_{\alpha}^{(1)}
<
X_{\alpha}^{(2)}
<
\cdots
<
X_{\alpha}^{(N)}.
\end{equation}
The corresponding canonically ordered monomials are
\begin{equation}
\left(X_{\alpha}^{(1)}\right)^{n_1}
\left(X_{\alpha}^{(2)}\right)^{n_2}
\cdots
\left(X_{\alpha}^{(N)}\right)^{n_N},
\qquad
n_1,\ldots,n_N\in\mathbb N_0.
\label{eq:PBWbasis}
\end{equation}
The PBW property means that every word in the generators can be
reordered, using the defining relations \eqref{eq:qimunu}
into a scalar multiple of a monomial of the form
\eqref{eq:PBWbasis}, and that these ordered monomials are linearly
independent. They therefore form a basis of $\mathcal S_\alpha$ as a
vector space. For the multiparameter quantum affine spaces considered
here, with non-zero exchange parameters satisfying $q_{\alpha}^{\mu\nu}q_{\alpha}^{\nu\mu}=1$, the ordered monomials \eqref{eq:PBWbasis} constitute the standard PBW basis. In particular, distinct canonically ordered monomials cannot cancel through additional algebraic relations. This follows, for example, from the standard Diamond Lemma normal-form argument \cite{Bergman1978Diamond,Oh2008QuantumPoisson}. This property allows us to compare the coefficients of the different ordered monomials independently when imposing the scalene Yang--Baxter relation.

\begin{theorem}[General phase-cancellation criterion]
\label{thm:phasecriterion}
Suppose that the canonically ordered monomials in each of the three
algebras $\mathcal S_\alpha$ form a linearly independent PBW-type set.
Then the operators \eqref{eq:Nedge} satisfy the scalene Yang--Baxter
relation
\begin{equation}
A_{12}B_{13}C_{23}
=
C_{23}B_{13}A_{12}
\label{eq:NSYB}
\end{equation}
for arbitrary independent coefficient families
$\{a_\mu\},\{b_\mu\},\{c_\mu\}$ if and only if
\begin{equation}
q_I^{\mu\nu}
q_{II}^{\mu\rho}
q_{III}^{\nu\rho}
=
1,
\qquad
\text{for all }
\mu,\nu,\rho\in\{1,\ldots,N\}.
\label{eq:phasecriterion}
\end{equation}
\end{theorem}

\begin{proof}
Expanding the left hand side of \eqref{eq:NSYB} gives
\begin{align}
A_{12}B_{13}C_{23}
={}&
\sum_{\mu,\nu,\rho=1}^{N}
a_\mu b_\nu c_\rho
\left(
X_{1;I}^{(\mu)}
X_{1;I}^{(\nu)}
\right)
\left(
X_{2;II}^{(\mu)}
X_{2;II}^{(\rho)}
\right)
\left(
X_{3;III}^{(\nu)}
X_{3;III}^{(\rho)}
\right).
\label{eq:Nlhs}
\end{align}
Here we have used the fact that operators acting on distinct lattice
sites commute, allowing every term to be grouped site by site.

The reverse product is
\begin{align}
C_{23}B_{13}A_{12}
={}&
\sum_{\mu,\nu,\rho=1}^{N}
a_\mu b_\nu c_\rho
\left(
X_{1;I}^{(\nu)}
X_{1;I}^{(\mu)}
\right)
\left(
X_{2;II}^{(\rho)}
X_{2;II}^{(\mu)}
\right)
\left(
X_{3;III}^{(\rho)}
X_{3;III}^{(\nu)}
\right).
\label{eq:Nrhs0}
\end{align}
Using the quadratic relations \eqref{eq:qimunu}, the local products can
be returned to the same canonical ordering as in \eqref{eq:Nlhs}.
This gives us
\begin{align}
C_{23}B_{13}A_{12}
={}
\sum_{\mu,\nu,\rho=1}^{N}
\frac{a_\mu b_\nu c_\rho}
{q_I^{\mu\nu}
 q_{II}^{\mu\rho}
 q_{III}^{\nu\rho}}
\left(
X_{1;I}^{(\mu)}
X_{1;I}^{(\nu)}
\right)
\left(
X_{2;II}^{(\mu)}
X_{2;II}^{(\rho)}
\right)
\left(
X_{3;III}^{(\nu)}
X_{3;III}^{(\rho)}
\right).
\label{eq:Nrhs}
\end{align}
Since the coefficient families are arbitrary and the canonically ordered monomials are linearly independent, equality of \eqref{eq:Nlhs} and \eqref{eq:Nrhs} requires equality in every $(\mu,\nu,\rho)$ sector. This is equivalent to \eqref{eq:phasecriterion}.
Conversely, if this condition is satisfied,
then every coefficient sector in the two cubic products agrees, proving
\eqref{eq:NSYB}.
\end{proof}
The phase-cancellation condition \eqref{eq:phasecriterion} can be solved completely.

\begin{theorem}[Classification of the multiparameter deformations]
\label{thm:classification}
Under the assumptions of Theorem~\ref{thm:phasecriterion}, the general
solution of \eqref{eq:phasecriterion} is $q_I^{\mu\nu}=q_{III}^{\mu\nu}
=q^{\mu\nu}$, $q_{II}^{\mu\nu}=\left(q^{\mu\nu}\right)^{-1}$.
The exchange parameters are necessarily of the form $q^{\mu\nu}=\frac{\lambda_\mu}{\lambda_\nu}$, with $\lambda_\mu\in\C^\times$. The parameters $\lambda_\mu$ are unique up to a common non-zero rescaling.
\end{theorem}
\begin{proof}
Setting $\rho=\nu$ in \eqref{eq:phasecriterion}, and using the fact that $q_{III}^{\nu\nu}=1$, one obtains $q_I^{\mu\nu}q_{II}^{\mu\nu}=1$
implying
\begin{equation}
q_{II}^{\mu\nu}
=
\left(q_I^{\mu\nu}\right)^{-1}.
\label{eq:q2q1}
\end{equation}
In the same way, setting $\mu=\nu$ in \eqref{eq:phasecriterion} and using the relation $q_I^{\mu\mu}=1$, gives $q_{II}^{\mu\rho}q_{III}^{\mu\rho}=1$. Combining these two we find
\begin{equation}
q_{III}^{\mu\rho}=\left(q_{II}^{\mu\rho}\right)^{-1}=q_I^{\mu\rho}.
\label{eq:q3q1}
\end{equation}
This proves the first part of the Theorem,
\begin{equation}
q^{\mu\nu}\equiv q_I^{\mu\nu}=q_{III}^{\mu\nu}=\left(q_{II}^{\mu\nu}\right)^{-1}.
\end{equation}
Substitution into the full phase criterion gives
\begin{equation}
q^{\mu\nu}
\left(q^{\mu\rho}\right)^{-1}
q^{\nu\rho}
=
1 ~~\Longrightarrow~~ q^{\mu\nu}q^{\nu\rho}
=
q^{\mu\rho}.
\label{eq:transitive}
\end{equation}
To solve \eqref{eq:transitive}, fix an arbitrary reference generator
index $\mu_0$ and define $\lambda_\mu=q^{\mu\mu_0}$.
Now setting $\rho=\mu_0$ in \eqref{eq:transitive} gives
\begin{equation}
q^{\mu\nu}q^{\nu\mu_0}
=
q^{\mu\mu_0}.
\end{equation}
Since all exchange parameters are non-zero, it follows that
\begin{equation}
q^{\mu\nu}
=
\frac{q^{\mu\mu_0}}{q^{\nu\mu_0}}
=
\frac{\lambda_\mu}{\lambda_\nu}.
\label{eq:qratio_general}
\end{equation}
Hence every solution of \eqref{eq:transitive} is necessarily of the
ratio form \eqref{eq:qratio_general}.
Conversely, if $q^{\mu\nu}=\frac{\lambda_\mu}{\lambda_\nu}$ then
\begin{equation}
q^{\mu\nu}q^{\nu\rho}
=
\frac{\lambda_\mu}{\lambda_\nu}
\frac{\lambda_\nu}{\lambda_\rho}
=
\frac{\lambda_\mu}{\lambda_\rho}
=
q^{\mu\rho},
\end{equation}
so the ratio form satisfies \eqref{eq:transitive}. This proves the second part of the Theorem.

Finally, this parametrization is unique up to a common non-zero
rescaling. Indeed, if
\begin{equation}
\frac{\lambda_\mu}{\lambda_\nu}
=
\frac{\widetilde{\lambda}_\mu}
     {\widetilde{\lambda}_\nu}
\end{equation}
for all $\mu,\nu$, then
\begin{equation}
\frac{\widetilde{\lambda}_\mu}{\lambda_\mu}
=
\frac{\widetilde{\lambda}_\nu}{\lambda_\nu}
\end{equation}
for all $\mu,\nu$. Hence there exists a single
$c\in\C^\times$ such that $\widetilde{\lambda}_\mu=c\,\lambda_\mu$
for every $\mu$.
\end{proof}

The condition \eqref{eq:transitive} may equivalently be written as
\begin{equation}
q^{\mu\nu}q^{\nu\rho}q^{\rho\mu}=1.
\label{eq:flatphase}
\end{equation}
Thus the product of the exchange factors around every triangle of
generator indices $(\mu,\nu,\rho)$ is trivial. In this sense the
scalene Yang--Baxter relation selects a ``flat'' multiparameter
deformation from the more general class of quantum affine spaces.

Combining Theorems~\ref{thm:phasecriterion} and
\ref{thm:classification} gives the principal result of the $N$-generator construction.

\begin{corollary}[$N$-generator scalene solution]
\label{cor:Nmain}
Let $\lambda_1,\ldots,\lambda_N\in\C^\times$.
In the algebras $\mathcal S_I$ and $\mathcal S_{III}$ impose
\begin{equation}
X_{j;\alpha}^{(\mu)}
X_{j;\alpha}^{(\nu)}
=
\frac{\lambda_\mu}{\lambda_\nu}
X_{j;\alpha}^{(\nu)}
X_{j;\alpha}^{(\mu)},
\qquad
\alpha\in\{I,III\},
\label{eq:outeralg}
\end{equation}
while in $\mathcal S_{II}$ impose the oppositely oriented relations
\begin{equation}
X_{j;II}^{(\mu)}
X_{j;II}^{(\nu)}
=
\frac{\lambda_\nu}{\lambda_\mu}
X_{j;II}^{(\nu)}
X_{j;II}^{(\mu)}.
\label{eq:middlealg}
\end{equation}
Then the operators $A_{12},B_{13},C_{23}$ defined in
\eqref{eq:Nedge} satisfy
\begin{equation}
A_{12}B_{13}C_{23}
=
C_{23}B_{13}A_{12}
\end{equation}
for arbitrary coefficient families
$\{a_\mu\},\{b_\mu\},\{c_\mu\}$.
\end{corollary}

The following observations further elaborate on these solutions.
\begin{enumerate}
    \item The construction contains $N$ non-zero parameters
$\lambda_1,\ldots,\lambda_N$, but only $N-1$ of them are independent. This is seen from the fact that the common rescaling $\lambda_\mu\longmapsto c\,\lambda_\mu$,
leaves every ratio $\lambda_\mu/\lambda_\nu$ invariant. The exchange
matrix selected by the scalene Yang--Baxter relation therefore contains
$N-1$ independent multiplicative deformation parameters. One may, for
example, fix this redundancy by choosing $\lambda_1=1$.
\item The $N=2$ case can now be seen as a special case of this construction. Choose $\lambda_1=1$ and $\lambda_2=q^{-1}$. Then $q^{12}
=\frac{\lambda_1}{\lambda_2}=q$,
so that the relations in $\mathcal S_I$, $\mathcal S_{II}$ and $\mathcal S_{III}$ reduce to exactly those found in Section \ref{sec:qdeform}.
The anticommuting case of
Section~\ref{sec:twogenerator} is recovered at $q=-1$.
\end{enumerate}
We close this section with the following remarks.
\begin{remark}
For $N=2$, the choice $q^{12}=-1$, produces an anticommuting pair at every site because $(-1)^{-1}=-1$.  This is the source of the simple Clifford/Pauli realization in Section~\ref{sec:twogenerator}. However, this argument fails for an arbitrary $N$. For $N\ge3$, one cannot choose $q^{\mu\nu}=-1$ for $\mu\neq\nu$ for all pairs. Indeed, for three distinct labels $\mu,\nu,\rho$, transitivity would require $(-1)(-1)=-1$, which is false. Thus a complete graph of pairwise anticommuting generator labels is incompatible with the generic $N$-generator matched-edge construction once three distinct generators are present. This implies that we cannot use this construction for a Clifford algebra with $N>2$ generators.

There is nevertheless a useful $\mathbb Z_2$ specialization. Let $\lambda_\mu\in\{+1,-1\}$. Then $q_{\mu\nu}=\frac{\lambda_\mu}{\lambda_\nu}\in\{+1,-1\}$. Generators in the same sign class commute, while generators in opposite sign classes anticommute. The anticommutation graph is therefore complete bipartite rather than complete. The two-generator Clifford construction is precisely the smallest nontrivial member of this bipartite pattern.  
\end{remark}
\begin{remark}
    A multiparameter quantum affine $N$-space is generated by $x_1,\ldots,x_N$ with relations
\begin{equation}
x_\mu x_\nu=q_{\mu\nu}x_\nu x_\mu,
\qquad
q_{\mu\mu}=1,
\qquad
q_{\mu\nu}q_{\nu\mu}=1.
\label{eq:qaffine}
\end{equation}
These algebras are standard examples of iterated skew-polynomial rings and admit ordered monomial bases under the usual nonvanishing parameter assumptions \cite{GoodearlLetzter2008,MukherjeeBera2024}.  Our construction does not use the most general multiplicatively antisymmetric matrix.  Instead the scalene YBE selects the rank-one multiplicative form $q_{\mu\nu}=\frac{\lambda_\mu}{\lambda_\nu}$.
In this sense the scalene relation imposes an exactness condition on the exchange phases.

The transitivity relation \eqref{eq:transitive} can be rewritten as
$q_{\mu\nu}q_{\nu\rho}q_{\rho\mu}=1$.
We may therefore view $q_{\mu\nu}$ as a multiplicative connection on the complete graph of generator labels, with trivial holonomy on every triangle.  The potentials $\lambda_\mu$ trivialize this connection.  This language is only an interpretation, but it captures why the parameter solution is so restrictive: arbitrary pairwise exchange phases are not allowed; they must be globally compatible with a single set of generator parameters.
\end{remark}
\begin{remark}
We now look at a simple semiclassical interpretation of the quantum affine space algebra. Write $\lambda_\mu=e^{\hbar\ell_\mu}$ and thus $q_{\mu\nu}=e^{\hbar(\ell_\mu-\ell_\nu)}$. Then
\begin{equation}
x_\mu x_\nu=e^{\hbar(\ell_\mu-\ell_\nu)}x_\nu x_\mu.
\end{equation}
Formally expanding around $\hbar=0$ gives the log-canonical Poisson bracket
\begin{equation}
\{x_\mu,x_\nu\}
=(\ell_\mu-\ell_\nu)x_\mu x_\nu,
\label{eq:Poisson}
\end{equation}
with the opposite sign at the middle site. Log-canonical Poisson structures arise naturally as semiclassical limits of quantum affine spaces \cite{GoodearlLetzter2008}. In the present case the antisymmetric coefficient matrix is itself exact, $\ell_{\mu\nu}=\ell_\mu-\ell_\nu$. This is the additive counterpart of the multiplicative flatness condition \eqref{eq:flatphase}.
\end{remark}


\section{Finite-dimensional representations}
\label{sec:finite}
We now study finite-dimensional matrix representations of the quantum
affine algebras introduced in Sections~\ref{sec:qdeform} and
\ref{sec:Ngenerator}. We take the local Hilbert space on site $j$ to be
$\mathcal H_j\simeq\C^d$ and represent the generators $X_{j;\alpha}^{(\mu)}$ by $d\times d$ matrices. The corresponding scalene operators $A_{12}$, $B_{13}$ and $C_{23}$ then act non-trivially on $\C^d\otimes\C^d$. An elementary determinant argument already places a strong restriction on finite-dimensional representations in which the local generators are invertible.

\begin{proposition}[Determinant obstruction]
\label{prop:det}
Let $X,Y\in\operatorname{Mat}_d(\C)$ be invertible matrices satisfying
$XY=qYX$ with $q\in\C^\times$. Then $q^d=1$.
\end{proposition}
\begin{proof}
The proof follows from the definitions. Taking determinants gives
\begin{equation*}
\det X\det Y(q^d-1) = 0.
\end{equation*}
As both determinants are non-zero, it follows that $q^d=1$.
\end{proof}
For the $N$-generator algebras of Section~\ref{sec:Ngenerator}, the
exchange parameters are $q^{\mu\nu} = \frac{\lambda_\mu}{\lambda_\nu}$.
Consequently, if all $N$ generators are represented by invertible
$d\times d$ matrices, Proposition~\ref{prop:det} implies $\left(
\frac{\lambda_\mu}{\lambda_\nu}\right)^d=1$ for every $\mu,\nu$.
Taking into account the common rescaling freedom of Theorem~\ref{thm:classification}, we can set $\lambda_1=1$ to obtain $\lambda_\mu^d=1$ for $\mu=1,\ldots,N$. Thus the $\lambda_\mu$ must be chosen from the $d$th roots of unity.
In particular, if all $\lambda_\mu$ are required to be distinct, then $N\leq d$. This observation naturally suggests local dimension $d$ for a fully distinct $d$-generator realization.

This restriction on the $\lambda$'s  applies only when both generators in a $q$-commuting pair are invertible. Singular local generators can evade this obstruction, while the resulting scalene operators can still be invertible. We consider both these types of representations, beginning with the latter first.

\subsection{Example : A two-dimensional singular realization}
\label{subsec:2dsingular}
For $q\in\C^\times$, consider the following matrices and the algebra generated by them,
\begin{equation}
N=
\begin{pmatrix}
0&1\\
0&0
\end{pmatrix}~;~
K_q=
\begin{pmatrix}
1&0\\
0&q
\end{pmatrix}~;~NK_q=qK_qN.
\label{eq:NKq}
\end{equation}
Using this we can construct the following representation of the three two-generator algebras of Section~\ref{sec:qdeform},
\begin{eqnarray}
& X_{j;\alpha}^{(1)}=N,~~X_{j;\alpha}^{(2)}=K_q~;~\alpha\in\{I,III\} & \nonumber  \\
& X_{j;II}^{(1)}=N,~~X_{j;II}^{(2)}=K_{q^{-1}} &
\label{eq:2drep}
\end{eqnarray}
The first generator in each algebra is nilpotent and singular, while the second generator is invertible. The three scalene operators become
\begin{equation}
\label{eq:2dscalene}
A_{12}=a\,N_1N_2+b\,(K_q)_1(K_{q^{-1}})_2~;~B_{13}=c\,N_1N_3+
d\,(K_q)_1(K_q)_3~;~C_{23}=e\,N_2N_3+f\,(K_{q^{-1}})_2 (K_q)_3,
\end{equation}
with explicit matrix forms given by
\begin{equation}
A=
\begin{pmatrix}
b & 0 & 0 & a\\
0 & bq^{-1} & 0 & 0\\
0 & 0 & bq & 0\\
0 & 0 & 0 & b
\end{pmatrix},
\qquad
B=
\begin{pmatrix}
d & 0 & 0 & c\\
0 & dq & 0 & 0\\
0 & 0 & dq & 0\\
0 & 0 & 0 & dq^2
\end{pmatrix},
\qquad
C=
\begin{pmatrix}
f & 0 & 0 & e\\
0 & fq & 0 & 0\\
0 & 0 & fq^{-1} & 0\\
0 & 0 & 0 & f
\end{pmatrix}.
\end{equation}
For generic non-zero values of $b$, $d$ and $f$ these three operators are invertible
even though one generator in each of the three local algebras is singular. Moreover, $q$ remains completely arbitrary. This example illustrates an important distinction: finite-dimensional invertibility of the scalene operators $A$, $B$ and $C$ does not require invertibility of every generator used in their construction.

\begin{remark}
It is also useful to verify that this singular realization is purely scalene for generic values of the deformation parameter. Defining the ordinary Yang--Baxter defect of an operator $R$ by
\begin{equation*}
\Delta_R=R_{12}R_{13}R_{23}-R_{23}R_{13}R_{12}.
\end{equation*}
For the three operators in the present realization, direct computation
gives
\begin{equation}
\Delta_A=0 \iff ab^2(q^2-1)=0~;~
\Delta_B=0 \iff cd^2(q^2-1)=0~;~
\Delta_C=0 \iff ef^2(q^2-1)=0.
\end{equation}
Consequently, when both terms in each scalene operator are present and
$q^2\neq1$, none of $A$, $B$ or $C$ satisfies the ordinary
non-braided constant Yang--Baxter equation individually, although the ordered
triple $(A,B,C)$ satisfies the scalene Yang--Baxter equation. The
special values $q=\pm1$ constitute exceptional points at which the
individual Yang--Baxter defects vanish.

For $q^2\neq1$ and generic coefficients, this realization also
obstructs a generic common spectral reparametrization. Indeed,
$B$ is diagonalizable with minimal polynomial
$(x-d)(x-dq)(x-dq^2)$, whereas $A$ and $C$ have nontrivial
Jordan blocks and minimal polynomials of degree four. These
degrees are invariant under invertible local similarities and
nonzero scalar rescalings. The condition
$\dim\operatorname{span}\{I,R,R^2,R^3\}\leq3$ is algebraic;
if satisfied by $R(u,v;\mu)$ for independent rapidities and fixed
shared moduli, it would extend throughout a connected meromorphic
family, contradicting the generic behavior of $A$ and $C$.
This leaves open exceptional spectral embeddings. At $q=\pm1$,
however, all three operators already belong to the common
Yang--Baxter family
$R(u)=\alpha(u)N\otimes N+\beta(u)K_q\otimes K_q$.

\end{remark}

\begin{remark}
    It is not necessary to use the same generators for $\mathcal{S}_I$ and $\mathcal{S}_{III}$. For instance the generators $X^{(1)}_{j;III}= N^T$ and $X^{(2)}_{j;III} = qK_{q^{-1}}$ also satisfy the right algebraic relations needed to obtain the non-braided scalene triple. However, the scalene triple obtained using these $\mathcal{S}_{III}$ generators is equivalent to those in \eqref{eq:2dscalene} by conjugation with $\sigma^x$ on the third tensor factor. 

    Note that we have used a broader notion of equivalence here when compared to the standard YBE case. It is important that a scalene solution is an \emph{ordered} triple $(A,B,C)$, with the three operators assigned respectively to the tensor pairs $(12)$, $(13)$ and $(23)$. Thus
\begin{equation*}
A_{12}B_{13}C_{23}
=
C_{23}B_{13}A_{12}
\end{equation*}
does not in general imply that a permutation of the operators, for
example $(B,C,A)$, satisfies the corresponding scalene relation with
the tensor positions $(12),(13),(23)$ kept fixed. Accordingly, the
natural local equivalence between two ordered scalene triples is induced by independent changes of basis in the three underlying spaces:
\begin{align}
A'&=(Q_I\otimes P_{II})A(Q_I\otimes P_{II})^{-1}, \nonumber\\
B'&=(Q_I\otimes L_{III})B(Q_I\otimes L_{III})^{-1},\nonumber \\
C'&=(P_{II}\otimes L_{III})C(P_{II}\otimes L_{III})^{-1}, \nonumber
\end{align}
where $Q_I,P_{II},L_{III}$ are invertible. A permutation of the three
underlying spaces constitutes a separate relabeling operation and
must be accompanied by the corresponding induced permutation, and
possibly reversal, of the tensor legs of the operators.
\end{remark}

\subsection{Finite Weyl representations at roots of unity}
\label{subsec:weyl}
The most symmetric finite-dimensional realizations are obtained from the finite Heisenberg--Weyl algebra. Let $\omega_d = \exp\left(\frac{2\pi\mathrm{i}}{d}\right)$ and introduce the $d\times d$ clock and shift matrices
\begin{equation}
U_d|n\rangle =\omega_d^n|n\rangle~;~V_d|n\rangle=|n+1\!\!\!\pmod d\rangle~;~U_dV_d=\omega_d V_dU_d
\label{eq:clockshift}
\end{equation}
for $n=0,\ldots,d-1$. Choosing $\lambda_\mu =\omega_d^{k_\mu}$ with $k_\mu\in\mathbb Z_d$, the algebras $\mathcal S_I$ and $\mathcal S_{III}$ are generated by
\begin{equation}
X_{j;\alpha}^{(\mu)}
=
U_dV_d^{-k_\mu},
\qquad
\alpha\in\{I,III\},
\label{eq:Weylouter}
\end{equation}
while for $\mathcal S_{II}$ define
\begin{equation}
X_{j;II}^{(\mu)}
=
U_dV_d^{k_\mu}.
\label{eq:Weylmiddle}
\end{equation}
The site index $j$ indicates the copy of the local matrix algebra on
which these matrices act.

Using these relations it is easy to see that 
\begin{equation}
\left(U_dV_d^{-k_\mu}\right)
\left(U_dV_d^{-k_\nu}\right)
=
\omega_d^{k_\mu-k_\nu}
\left(U_dV_d^{-k_\nu}\right)
\left(U_dV_d^{-k_\mu}\right),
\end{equation}
which implies
\begin{equation}
X_{j;\alpha}^{(\mu)}
X_{j;\alpha}^{(\nu)}
=
\frac{\lambda_\mu}{\lambda_\nu}
X_{j;\alpha}^{(\nu)}
X_{j;\alpha}^{(\mu)},
\qquad
\alpha\in\{I,III\},
\label{eq:Weylouterrelation}
\end{equation}
and 
\begin{equation}
X_{j;II}^{(\mu)}
X_{j;II}^{(\nu)}
=
\frac{\lambda_\nu}{\lambda_\mu}
X_{j;II}^{(\nu)}
X_{j;II}^{(\mu)}.
\label{eq:Weylmiddlerelation}
\end{equation}
Thus the finite Weyl algebra gives precisely the oppositely oriented
quantum affine spaces required by Corollary~\ref{cor:Nmain}.  Every local
generator in \eqref{eq:Weylouter} and \eqref{eq:Weylmiddle} is unitary
and therefore invertible.

\subsubsection*{Three-generator representation in local dimension three}
\label{subsubsec:3gen}
The first fully distinct three-generator realization occurs for $d=N=3$ with $\omega=\exp\left(\frac{2\pi\mathrm{i}}{3}\right)$. 
The clock and shift matrices are
\begin{equation}
U_3=
\begin{pmatrix}
1&0&0\\
0&\omega&0\\
0&0&\omega^2
\end{pmatrix}~;~ V_3=
\begin{pmatrix}
0&0&1\\
1&0&0\\
0&1&0
\end{pmatrix}~;~U_3V_3=\omega V_3U_3.
\label{eq:U3V3}
\end{equation}
Choose $\lambda_1=1$, $\lambda_2=\omega$ and $\lambda_3=\omega^2$ 
implying $(k_1,k_2,k_3)=(0,1,2)$. Using this the representations of $\mathcal S_\alpha$ for $\alpha\in\{I,III\}$ are
\begin{equation}
X_{j;\alpha}^{(1)}=U_3~;~
X_{j;\alpha}^{(2)}=U_3V_3^{-1}~;~
X_{j;\alpha}^{(3)}=U_3V_3^{-2}.
\end{equation}
while the representation of $\mathcal S_{II}$ is
\begin{equation}
X_{j;II}^{(1)}=U_3~;~
X_{j;II}^{(2)}=U_3V_3~;~
X_{j;II}^{(3)}=U_3V_3^2.
\label{eq:3middlerep}
\end{equation}
All three generators are unitary, invertible and linearly independent.
The associated scalene operators
\begin{equation}
A_{12}=
\sum_{\mu=1}^{3}
a_\mu
X_{1;I}^{(\mu)}
X_{2;II}^{(\mu)}~;~
B_{13} =
\sum_{\mu=1}^{3}
b_\mu
X_{1;I}^{(\mu)}
X_{3;III}^{(\mu)}~;~
C_{23}=
\sum_{\mu=1}^{3}
c_\mu
X_{2;II}^{(\mu)}
X_{3;III}^{(\mu)}
\end{equation}
act on $\C^3\otimes\C^3\simeq\C^9$. They satisfy the scalene Yang--Baxter relation by Corollary~\ref{cor:Nmain}. Their explicit $9\times9$
matrix forms are  
\begin{equation}
\resizebox{0.60\textwidth}{!}{$
A=
\begin{pmatrix}
a_1 & 0 & 0 & 0 & 0 & a_2 & 0 & a_3 & 0 \\
0 & a_1\omega & 0 & a_2\omega & 0 & 0 & 0 & 0 & a_3\omega \\
0 & 0 & a_1\omega^2 & 0 & a_2\omega^2 & 0 & a_3\omega^2 & 0 & 0 \\
0 & a_3\omega & 0 & a_1\omega & 0 & 0 & 0 & 0 & a_2\omega \\
0 & 0 & a_3\omega^2 & 0 & a_1\omega^2 & 0 & a_2\omega^2 & 0 & 0 \\
a_3 & 0 & 0 & 0 & 0 & a_1 & 0 & a_2 & 0 \\
0 & 0 & a_2\omega^2 & 0 & a_3\omega^2 & 0 & a_1\omega^2 & 0 & 0 \\
a_2 & 0 & 0 & 0 & 0 & a_3 & 0 & a_1 & 0 \\
0 & a_2\omega & 0 & a_3\omega & 0 & 0 & 0 & 0 & a_1\omega
\end{pmatrix}.
$}
\label{eq:A9explicit}
\end{equation}
\begin{equation}
\resizebox{0.65\textwidth}{!}{$
B=
\begin{pmatrix}
b_1 & 0 & 0 & 0 & b_2 & 0 & 0 & 0 & b_3 \\
0 & b_1\omega & 0 & 0 & 0 & b_2\omega & b_3\omega & 0 & 0 \\
0 & 0 & b_1\omega^2 & b_2\omega^2 & 0 & 0 & 0 & b_3\omega^2 & 0 \\
0 & 0 & b_3\omega & b_1\omega & 0 & 0 & 0 & b_2\omega & 0 \\
b_3\omega^2 & 0 & 0 & 0 & b_1\omega^2 & 0 & 0 & 0 & b_2\omega^2 \\
0 & b_3 & 0 & 0 & 0 & b_1 & b_2 & 0 & 0 \\
0 & b_2\omega^2 & 0 & 0 & 0 & b_3\omega^2 & b_1\omega^2 & 0 & 0 \\
0 & 0 & b_2 & b_3 & 0 & 0 & 0 & b_1 & 0 \\
b_2\omega & 0 & 0 & 0 & b_3\omega & 0 & 0 & 0 & b_1\omega
\end{pmatrix}.
$}
\label{eq:B9explicit}
\end{equation}
\begin{equation}
\resizebox{0.60\textwidth}{!}{$
C=
\begin{pmatrix}
c_1 & 0 & 0 & 0 & 0 & c_3 & 0 & c_2 & 0 \\
0 & c_1\omega & 0 & c_3\omega & 0 & 0 & 0 & 0 & c_2\omega \\
0 & 0 & c_1\omega^2 & 0 & c_3\omega^2 & 0 & c_2\omega^2 & 0 & 0 \\
0 & c_2\omega & 0 & c_1\omega & 0 & 0 & 0 & 0 & c_3\omega \\
0 & 0 & c_2\omega^2 & 0 & c_1\omega^2 & 0 & c_3\omega^2 & 0 & 0 \\
c_2 & 0 & 0 & 0 & 0 & c_1 & 0 & c_3 & 0 \\
0 & 0 & c_3\omega^2 & 0 & c_2\omega^2 & 0 & c_1\omega^2 & 0 & 0 \\
c_3 & 0 & 0 & 0 & 0 & c_2 & 0 & c_1 & 0 \\
0 & c_3\omega & 0 & c_2\omega & 0 & 0 & 0 & 0 & c_1\omega
\end{pmatrix}.
$}
\label{eq:C9explicit}
\end{equation}
Furthermore, the determinants of each of the operators in the scalene triple are polynomials in the corresponding coefficient family. These
polynomials are not identically zero: for example, setting
$a_2=a_3=0$ with $a_1\neq0$ gives $A_{12}=a_1X_{1;I}^{(1)}X_{2;II}^{(1)}$
which is invertible. Thus the three $9\times9$ scalene operators are
invertible for generic values of their coefficients.

\subsubsection*{Four-generator representation in local dimension four}
\label{subsubsec:4gen}
The analogous fully distinct four-generator realization is obtained by
taking $d=N=4$ and $\omega_4 = \exp\left(\frac{2\pi\mathrm{i}}{4}\right)= \mathrm{i}$. Choose $(\lambda_1,\lambda_2,\lambda_3,\lambda_4)= (1,\mathrm{i},-1,-\mathrm{i})$. Then the clock and shift matrices can be written as
\begin{equation}
U_4=
\begin{pmatrix}
1&0&0&0\\
0&\mathrm{i}&0&0\\
0&0&-1&0\\
0&0&0&-\mathrm{i}
\end{pmatrix},
\qquad
V_4=
\begin{pmatrix}
0&0&0&1\\
1&0&0&0\\
0&1&0&0\\
0&0&1&0
\end{pmatrix},
\label{eq:U4V4}
\end{equation}
such that
\begin{equation}
U_4V_4
=
\mathrm{i}V_4U_4.
\end{equation}
Then choosing $(k_1,k_2,k_3,k_4)=(0,1,2,3)$, the representations of $\mathcal S_I$ and $\mathcal S_{III}$ become
\begin{equation}
X_{j;\alpha}^{(\mu)}
=
U_4V_4^{-k_\mu},
\qquad
\alpha\in\{I,III\},
\qquad
\mu=1,\ldots,4,
\label{eq:4outerrep}
\end{equation}
while the representation of $\mathcal{S}_{II}$ becomes,
\begin{equation}
X_{j;II}^{(\mu)}
=
U_4V_4^{k_\mu},
\qquad
\mu=1,\ldots,4.
\label{eq:4middlerep}
\end{equation}
The four matrices are unitary, invertible and linearly independent, and
all six pairwise exchange relations are represented nontrivially. The corresponding scalene operators act on $\C^4\otimes\C^4\simeq\C^{16}$. For example,
\begin{equation}
A_{12}
=
\sum_{\mu=1}^{4}
a_\mu
X_{1;I}^{(\mu)}
X_{2;II}^{(\mu)},
\label{eq:A4gen}
\end{equation}
with analogous expressions for $B_{13}$ and $C_{23}$.  By
Corollary~\ref{cor:Nmain}, this satisfies the scalene YBE.
As in the three-generator case, the determinants of the three
$16\times16$ scalene operators are non-zero polynomials in their
respective coefficients.  Since a specialization to any single
non-zero generator contribution gives an invertible tensor product,
the three scalene operators are generically invertible.

\begin{remark}
The finite-dimensional results exhibit a natural hierarchy.  If all
local generators are required to be invertible, their exchange factors
must be roots of unity.  A fully distinct $N$-generator realization
therefore requires at least $d\geq N$. For $N=2$, the choice $d=2$ and $q=-1$ gives the familiar Pauli anticommutation relation.  The cases $N=d=3$ and $N=d=4$ are its qutrit and four-state finite Heisenberg--Weyl analogues.  In this sense the original anticommuting construction is the first member of a finite-dimensional Weyl hierarchy associated with the scalene Yang--Baxter relation.
\end{remark}

\begin{remark}
The finite Weyl representations are generically purely scalene.
This is again seen by evaluating the ordinary Yang--Baxter defect of a two-site operator $R$,
\begin{equation*}
\Delta_R
=
R_{12}R_{13}R_{23}
-
R_{23}R_{13}R_{12}.
\end{equation*}
For the three-generator realization in local dimension three, and
assuming all coefficients are non-zero, direct calculation gives
\begin{equation*}
\Delta_A=0
\iff
a_1=a_2=a_3,
\qquad
\Delta_C=0
\iff
c_1=c_2=c_3,
\end{equation*}
whereas
\begin{equation*}
\Delta_B\neq0
\end{equation*}
whenever $b_1b_2b_3\neq0$. Thus for generic coefficients none of the
three scalene operators satisfies the ordinary Yang--Baxter equation.

For the four-generator realization in local dimension four, assuming
again that all coefficients are non-zero, one finds
\begin{align*}
\Delta_A=0
&\iff
a_1=a_3,\quad a_2=a_4,
\\
\Delta_B=0
&\iff
b_1=b_3,\quad b_2=b_4,
\\
\Delta_C=0
&\iff
c_1=c_3,\quad c_2=c_4.
\end{align*}
Consequently, the four-generator construction is also purely scalene
for generic values of its coefficients, although special
lower-dimensional parameter subspaces exist on which the individual
operators become ordinary constant Yang--Baxter solutions.

Moreover, both examples admit an obstruction to a generic common spectral reparametrization. Indeed, for the three-generator case we have
\[
B^3=\bigl(b_1^3+b_2^3+b_3^3-3b_1b_2b_3\bigr)I_9,
\]
whereas $A^3$ and $C^3$ are generically nonscalar. These properties
are preserved by invertible local similarities and nonzero scalar
rescalings. Suppose that, up to these transformations, the triple
could be represented as
$R(u,v;\mu),R(u,w;\mu),R(v,w;\mu)$ for a common meromorphic Yang--Baxter
family, with independently varying $u,v,w$ and fixed shared moduli
$\mu$. The scalar-cube identity for $B$ would then imply, by analytic
continuation on a connected spectral domain, that
$R(x,y;\mu)^3\propto I_9$ identically, contradicting the generic
behavior of $A$ and $C$. This excludes such a generic
reparametrization, while leaving open embeddings at specially
constrained spectral values.

The four generator case also obstructs a generic common spectral
reparametrization. Indeed, $B$ commutes with
$U_4\otimes U_4^{-1}$ and $V_4^{-1}\otimes V_4$,
which anticommute with each other. Consequently, every eigenvalue of $B$
has even algebraic multiplicity, whereas $A$ and $C$ generically
have simple spectra. This distinction is preserved by invertible
local similarities and nonzero scalar rescalings. A common
meromorphic parametrization
$R(u,v;\mu),R(u,w;\mu),R(v,w;\mu)$ with independent rapidities and
fixed shared moduli would therefore force the characteristic-polynomial
discriminant of $R(x,y;\mu)$ to vanish identically on its connected
spectral domain, contradicting the generic spectra of $A$ and $C$.
As before, this argument leaves open embeddings at specially
constrained spectral values.

\end{remark}

\section{Infinite-dimensional representations}
\label{sec:infinite}

The finite-dimensional representations discussed in
Section~\ref{sec:finite} impose root-of-unity restrictions on the
deformation parameters whenever all the local generators are required
to be invertible.  These restrictions disappear in infinite-dimensional
representations.  In fact, the multiparameter quantum affine algebras
of Section~\ref{sec:Ngenerator} admit a particularly simple realization
in terms of weighted-shift operators for completely arbitrary non-zero
values of the parameters $\lambda_\mu$.

\subsection{Bilateral weighted-shift representations}
\label{subsec:weightedshifts}

We begin by explaining the space on which the infinite-dimensional
representation acts and the meaning of a bilateral weighted-shift
operator. Let
\begin{equation}
\mathcal D
=
c_{00}(\mathbb Z)
=
\operatorname{span}_{\C}
\left\{
|n\rangle:n\in\mathbb Z
\right\}
\label{eq:Dspace}
\end{equation}
denote the vector space of finite-support sequences indexed by the
integers\footnote{The notation $c_{00}(\mathbb Z)$ denotes the vector space of finitely supported complex sequences indexed by $\mathbb Z$,
\begin{equation*}
c_{00}(\mathbb Z) = \left\{(\psi_n)_{n\in\mathbb Z}: \psi_n\in\C,\; \psi_n=0
\text{ for all but finitely many }n
\right\}.
\end{equation*}
Equivalently, $c_{00}(\mathbb Z) = \bigoplus_{n\in\mathbb Z}\C\,|n\rangle$, so every vector has the form $|\psi\rangle =
\sum_{n\in\mathbb Z}\psi_n|n\rangle$ with only finitely many non-zero coefficients $\psi_n$. The support of
such a sequence, $\operatorname{supp}\psi = \{n\in\mathbb Z:\psi_n\neq0\}$, is therefore finite. The space $c_{00}(\mathbb Z)$ is a dense subspace of $\ell^2(\mathbb Z)$ and provides a convenient common algebraic domain for the weighted-shift operators considered here, including the case in which they are unbounded on $\ell^2(\mathbb Z)$. See \cite{Vaidyanathan2023FunctionalAnalysis} for more details.}. Explicitly, an element of $\mathcal D$ has the form
\begin{equation}
|\psi\rangle
=
\sum_{n\in\mathbb Z}\psi_n|n\rangle,
\label{eq:finite_support_vector}
\end{equation}
where only finitely many coefficients $\psi_n$ are non-zero.  Thus,
although the label $n$ ranges over all of $\mathbb Z$, every individual
vector in $\mathcal D$ is a finite linear combination of basis vectors. The finite-support condition is useful because it provides a common algebraic domain on which all operators introduced below, as well as their products, are well defined. In particular, even when the
weights of a shift operator become arbitrarily large as $n\rightarrow\pm\infty$, the action of the operator on a vector in
$\mathcal D$ still produces only a finite linear combination of basis
vectors. The space $\mathcal D$ may subsequently be regarded as a
dense subspace of the Hilbert space $\ell^2(\mathbb Z)$ whenever
Hilbert-space properties are required.

The basis $\{|n\rangle:n\in\mathbb Z\}$ resembles the number-state
basis of a harmonic oscillator, but there is an important difference. For the usual bosonic Fock space one has
\begin{equation}
\mathsf N|n\rangle=n|n\rangle,
\qquad
n\in\mathbb N_0,
\end{equation}
so that the spectrum of the number operator is bounded from below and
$|0\rangle$ is a distinguished vacuum state. By contrast, in
\eqref{eq:Dspace} the integer label is unbounded in both directions, $n\in\mathbb Z$, and there is no lowest-weight state. The representation is therefore Fock-space-like, but it is not the usual bosonic Fock representation.

This distinction can also be expressed directly in oscillator
language. For the ordinary bosonic oscillator we have $[a,a^\dagger]=1$, and $\mathsf N=a^\dagger a$. Then the normalized creation operator $E^\dagger = a^\dagger(\mathsf N+1)^{-1/2}$ acts as a unilateral shift
\begin{equation}
E^\dagger|n\rangle=|n+1\rangle,
\qquad
n\in\mathbb N_0,
\end{equation}
moving $n$ in the set of natural numbers alone. However, such an operator is not invertible, because the Fock space possesses the lower boundary $|0\rangle$. A bilateral analogue can be defined by introducing an integer-valued number operator $\mathsf N$ and a shift operator $S$ satisfying
\begin{equation}
\mathsf N|n\rangle=n|n\rangle,
\qquad
S|n\rangle=|n+1\rangle,
\qquad
n\in\mathbb Z.
\label{eq:NSaction}
\end{equation}
Since there is no lowest-weight state, $S$ now has an inverse given by
\begin{equation}
S^{-1}|n\rangle=|n-1\rangle.
\end{equation}
Using these definitions it is easily seen that
\begin{equation}
[\mathsf N,S]=S.
\label{eq:NSalgebra}
\end{equation}
Thus the algebra relevant to the bilateral representation is more
naturally viewed as a number--shift, or quantum-rotor algebra \cite{OhnukiKitakado1993,CarruthersNieto1968} rather
than the ordinary harmonic-oscillator algebra.

The next ingredient in constructing this representation is the \emph{weighted-shift operator} $W$. It is defined by the action on the
above basis,
\begin{equation}
W|n\rangle = w_n|n+1\rangle;~~\{w_n:n\in\mathbb Z\}.
\label{eq:generalweightedshift}
\end{equation}
Here $w_n$ are called the sequence of weights. The shift is called
\emph{bilateral} because the basis is indexed by all integers,
$n\in\mathbb Z$, so that the shift extends in both directions along
the integer lattice. This is to be contrasted with the unilateral
shift of the usual Fock representation, where $n\in\mathbb N_0$ and a boundary is present at $n=0$.

If $w_n\neq0$ for every $n\in\mathbb Z$, then the bilateral weighted shift is algebraically invertible on $\mathcal D$, with
\begin{equation}
W^{-1}|n\rangle = \frac{1}{w_{n-1}}|n-1\rangle.
\label{eq:generalweightedinverse}
\end{equation}
The absence of a lower boundary and the resulting possibility of
invertibility are particularly useful for the present construction.

Next we introduce the weighted shifts required for the quantum affine
algebras. For every $\lambda\in\C^\times$, define the diagonal
operator
\begin{equation}
D_\lambda|n\rangle
=
\lambda^n|n\rangle.
\label{eq:Dlambda}
\end{equation}
Equivalently, in terms of the integer-valued number operator, $D_\lambda=\lambda^{\mathsf N}$. The relation \eqref{eq:NSalgebra} implies
\begin{equation}
D_\lambda S = \lambda S D_\lambda.
\label{eq:DSrelation}
\end{equation}
At lattice site $j$, let $\mathcal D_j$ be a copy of
$c_{00}(\mathbb Z)$, with basis $\left\{ |n\rangle_j:n\in\mathbb Z \right\}$. We define
\begin{equation}
G_{\lambda,j}
=
S_jD_{\lambda,j}
=
S_j\lambda^{\mathsf N_j}.
\label{eq:GSD}
\end{equation}
Its action is
\begin{equation}
G_{\lambda,j}|n\rangle_j
=
\lambda^n|n+1\rangle_j.
\label{eq:Gdef}
\end{equation}
Thus $G_{\lambda,j}$ is a bilateral weighted shift with weight
sequence $w_n=\lambda^n$. Then for any $\lambda,\mu\in\C^\times$, the above action implies
\begin{equation}
G_{\lambda,j}G_{\mu,j}=\frac{\lambda}{\mu}G_{\mu,j}G_{\lambda,j}.
\label{eq:Grelation}
\end{equation}
Then the $N$-parameter quantum affine algebras $\mathcal S_I$ and $\mathcal S_{III}$ of Corollary~\ref{cor:Nmain} are therefore represented by
\begin{equation}
X_{j;\alpha}^{(\mu)} = G_{\lambda_\mu,j}, \qquad
\alpha\in\{I,III\}, \qquad
\mu=1,\ldots,N.
\label{eq:Gouter}
\end{equation}
Using \eqref{eq:Grelation} we find
\begin{equation}
X_{j;\alpha}^{(\mu)}X_{j;\alpha}^{(\nu)}= \frac{\lambda_\mu}{\lambda_\nu}X_{j;\alpha}^{(\nu)}X_{j;\alpha}^{(\mu)},
\qquad \alpha\in\{I,III\},
\label{eq:Gouterrelation}
\end{equation}
as required. To realize the oppositely oriented algebra $\mathcal S_{II}$, define a second family of bilateral weighted shifts by
\begin{equation}
H_{\lambda,j}=S_jD_{\lambda^{-1},j}=S_j\lambda^{-\mathsf N_j},
\label{eq:HSD}
\end{equation}
with the action
\begin{equation}
H_{\lambda,j}|n\rangle_j = \lambda^{-n}|n+1\rangle_j.
\label{eq:Hdef}
\end{equation}
The corresponding sequence of weights is therefore $\widetilde w_n=\lambda^{-n}$. A direct calculation gives
\begin{equation}
H_{\lambda,j}H_{\mu,j} =\frac{\mu}{\lambda} H_{\mu,j}H_{\lambda,j}.
\label{eq:Hrelation}
\end{equation}
Thus a representation of the $N$-parameter $\mathcal S_{II}$ is obtained by setting
\begin{equation}
X_{j;II}^{(\mu)} = H_{\lambda_\mu,j}, \qquad \mu=1,\ldots,N,
\label{eq:Hmiddle}
\end{equation}
for which
\begin{equation}
X_{j;II}^{(\mu)} X_{j;II}^{(\nu)}=\frac{\lambda_\nu}{\lambda_\mu}
X_{j;II}^{(\nu)} X_{j;II}^{(\mu)}.
\label{eq:Hmiddlerelation}
\end{equation}
Since $\lambda\neq0$, none of the weights of either
$G_{\lambda,j}$ or $H_{\lambda,j}$ vanishes.  Both families are
therefore algebraically invertible on $\mathcal D_j$. Their inverses
act as
\begin{equation}
G_{\lambda,j}^{-1}|n\rangle_j
=
\lambda^{-(n-1)}|n-1\rangle_j~;~H_{\lambda,j}^{-1}|n\rangle_j
=
\lambda^{n-1}|n-1\rangle_j.
\label{eq:Ginverse}
\end{equation}
The distinction between the algebraic space $\mathcal D_j$ and its
Hilbert-space completion is important. On $\mathcal D_j$, all the
relations and inverse formulas above are valid for every
$\lambda\in\C^\times$. If $|\lambda|=1$, the weights have unit modulus,
$|\lambda^n|=1$,and both $G_{\lambda,j}$ and $H_{\lambda,j}$ extend to unitary operators on $\ell^2(\mathbb Z)$. If on the other hand, $|\lambda|\neq1$, then the weights grow without bound in one direction of the integer lattice, so the corresponding operators are generally unbounded on $\ell^2(\mathbb Z)$. Nevertheless, $\mathcal D_j$ is invariant under $G_{\lambda,j}$, $H_{\lambda,j}$ and their inverses, and therefore provides a natural common dense algebraic domain on which all the quantum-affine relations are exact.

\begin{theorem}[Infinite-dimensional realization of the $N$-generator scalene triple]
\label{thm:weighted}
Let $\lambda_1,\ldots,\lambda_N\in\C^\times$ be arbitrary complex numbers. On the three lattice sites appearing in the scalene
relation, represent the generators by $X_{1;I}^{(\mu)} = G_{\lambda_\mu,1}$, $X_{2;II}^{(\mu)} = H_{\lambda_\mu,2}$, $X_{3;III}^{(\mu)} = G_{\lambda_\mu,3}$. Then the scalene operators
\begin{eqnarray}
A_{12} =
\sum_{\mu=1}^{N}
a_\mu\,
G_{\lambda_\mu,1}
H_{\lambda_\mu,2}~;~ B_{13} =
\sum_{\mu=1}^{N}
b_\mu\,
G_{\lambda_\mu,1}
G_{\lambda_\mu,3}~;~C_{23} =
\sum_{\mu=1}^{N}
c_\mu\,
H_{\lambda_\mu,2}
G_{\lambda_\mu,3}
\label{eq:weightedA}
\end{eqnarray}
satisfy the scalene YBE $A_{12}B_{13}C_{23} =
C_{23}B_{13}A_{12}$, on the common invariant algebraic domain $\mathcal D_1\otimes\mathcal D_2\otimes\mathcal D_3$,
for arbitrary coefficient families $\{a_\mu\}$, $\{b_\mu\}$ and $\{c_\mu\}$.
\end{theorem}

\begin{proof}
Equations~\eqref{eq:Grelation} and \eqref{eq:Hrelation} reproduce
precisely the defining relations \eqref{eq:outeralg} and
\eqref{eq:middlealg} of the quantum affine algebras
$\mathcal S_I$, $\mathcal S_{II}$ and $\mathcal S_{III}$.
The result therefore follows directly from
Corollary~\ref{cor:Nmain}.
\end{proof}
Unlike the fully invertible finite-dimensional representations of
Section~\ref{sec:finite}, the construction of
Theorem~\ref{thm:weighted} imposes no root-of-unity condition on the
ratios $\frac{\lambda_\mu}{\lambda_\nu}$. Thus the full $(N-1)$-parameter family of quantum affine algebras selected by the scalene Yang--Baxter relation admits an explicit infinite-dimensional representation.

\subsection{Multiplicative-shift and $q$-difference realizations on function spaces}
\label{subsec:qdifference}
The bilateral weighted-shift representation introduced in
Subsection~\ref{subsec:weightedshifts} has an equivalent realization on
a space of functions. The natural algebraic space for this purpose is
the space of Laurent polynomials
\begin{equation}
\C[x,x^{-1}]
=
\left\{
\sum_{n=n_-}^{n_+} c_n x^n
\;\middle|\;
n_-,n_+\in\mathbb Z,\;
c_n\in\C
\right\}.
\end{equation}
The correspondence $|n\rangle\longleftrightarrow x^n$, identifies the basis of the bilateral representation with the Laurent
monomial basis. The use of Laurent polynomials, rather than ordinary
polynomials, is essential here because both positive and negative
integer powers are required in order to reproduce the bilateral basis
indexed by $n\in\mathbb Z$.

Let $M$ denote multiplication by the coordinate $x$, and let
$T_\lambda$, for $\lambda\in\C^\times$, denote the multiplicative-shift
or dilation operator
\begin{equation}
(Mf)(x)=xf(x)~;~(T_\lambda f)(x)=f(\lambda x).
\label{eq:MTdef}
\end{equation}
Their noncommutativity follows immediately from their action on an
arbitrary Laurent polynomial $f$. We have,
\begin{equation*}
(T_\lambda Mf)(x) =(Mf)(\lambda x) = \lambda x f(\lambda x)~;~(MT_\lambda f)(x) = x f(\lambda x),
\end{equation*}
implying
\begin{equation}
T_\lambda M = \lambda M T_\lambda.
\label{eq:TMrelation}
\end{equation}
Thus $T_\lambda$ and $M$ themselves furnish the standard quantum-plane
relation with deformation parameter $\lambda$. Now defining
\begin{equation}
G_\lambda = MT_\lambda,
\label{eq:MTG}
\end{equation}
we find that $G_\lambda G_\eta = \frac{\lambda}{\eta}G_\eta G_\lambda$, the required algebra. This proof uses the fact that the dilations form an Abelian group.
Thus these operators generate the quantum affine algebras $\mathcal S_I$ and $\mathcal S_{III}$.


At this point we can also see the equivalence of this representation with the bilateral weighted-shift representation. This is seen through the action on a Laurent monomial. Since $T_\lambda x^n = (\lambda x)^n =\lambda^n x^n$, one has $G_\lambda x^n =MT_\lambda x^n=\lambda^n x^{n+1}.$ Under the correspondence $|n\rangle\longleftrightarrow x^n$, this is identical to $G_\lambda|n\rangle = \lambda^n|n+1\rangle$ in \eqref{eq:Gdef}.

In a similar manner we can obtain the representation for $\mathcal S_{II}$ by defining
\begin{equation}
H_\lambda
=
MT_{\lambda^{-1}}.
\label{eq:MTH}
\end{equation}
With the action $H_\lambda x^n = \lambda^{-n}x^{n+1}$ on a Laurent monomial, they satisfy $H_\lambda H_\eta = \frac{\eta}{\lambda} H_\eta H_\lambda$ as required. The relation of this representation to $q$-difference calculus is shown in Appendix \ref{app:q-calculus}.


This function space representation can be more directly related to the number basis representation of Sec. \ref{subsec:weightedshifts}. The Euler operator $x\partial_x$, acts diagonally on Laurent monomials,
$x\partial_x\,x^n = n x^n$ and therefore plays the role of the integer-valued number operator $\mathsf N$. Consequently the dilation operator can be written as $T_\lambda = \lambda^{x\partial_x}$ or as $T_\lambda =\exp\left[(\log\lambda)x\partial_x\right]$. On Laurent polynomials this expression is independent of the choice of branch of $\log\lambda$, since the eigenvalues of $x\partial_x$ are
integers. From this it follows that
\begin{eqnarray}
& G_\lambda = x\,\lambda^{x\partial_x} = x\, \exp\left[(\log\lambda)x\partial_x \right], & \nonumber \\
& H_\lambda = x\,\lambda^{-x\partial_x} = x\, \exp\left[ -(\log\lambda)x\partial_x\right]. &
\label{eq:Hdilation}
\end{eqnarray}
The correspondence between the two infinite-dimensional realizations
can therefore be summarized as
\begin{equation}
|n\rangle
\longleftrightarrow
x^n,
\qquad
\mathsf N
\longleftrightarrow
x\partial_x,
\qquad
S
\longleftrightarrow
M,
\label{eq:representationcorrespondence}
\end{equation}
under which
\begin{equation}
S\lambda^{\mathsf N}
\longleftrightarrow
M\lambda^{x\partial_x}.
\end{equation}
Thus the bilateral weighted-shift and function-space constructions are
two equivalent realizations of the same quantum affine algebraic
structure.  Operator representations of quantum planes and their
localizations have been studied extensively; see, for example,
Ref.~\cite{Schmudgen2002}.

\begin{remark}[Generic scalene character of the infinite-dimensional realization]
The infinite-dimensional weighted-shift realization is also purely
scalene for generic parameters. Let
\begin{equation*}
\mathcal A(k)=\sum_{\mu=1}^N a_\mu\lambda_\mu^k,
\qquad
\mathcal B(k)=\sum_{\mu=1}^N b_\mu\lambda_\mu^k,
\qquad
\mathcal C(k)=\sum_{\mu=1}^N c_\mu\lambda_\mu^k.
\end{equation*}
For the ordinary Yang--Baxter defect
\begin{equation*}
\Delta_R
=
R_{12}R_{13}R_{23}
-
R_{23}R_{13}R_{12},
\end{equation*}
one finds on a basis state $|m,n,\ell\rangle$
\begin{align}
\Delta_A|m,n,\ell\rangle
={}&
\mathcal A(n-\ell)\mathcal A(m-n)
\left[
\mathcal A(m-\ell-1)
-
\mathcal A(m-\ell+1)
\right]
\nonumber\\
&\times|m+2,n+2,\ell+2\rangle,
\\
\Delta_B|m,n,\ell\rangle
={}&
\mathcal B(m+\ell+1)
\Big[
\mathcal B(n+\ell)\mathcal B(m+n+2)
\nonumber\\
&\hspace{25mm}
-
\mathcal B(m+n)\mathcal B(n+\ell+2)
\Big]
|m+2,n+2,\ell+2\rangle,
\\
\Delta_C|m,n,\ell\rangle
={}&
\mathcal C(\ell-n)\mathcal C(n-m)
\left[
\mathcal C(\ell-m+1)
-
\mathcal C(\ell-m-1)
\right]
\nonumber\\
&\times|m+2,n+2,\ell+2\rangle.
\end{align}
Consequently, for generically non-vanishing weights, the individual
Yang--Baxter equations require
\begin{equation}
\mathcal A(k+2)=\mathcal A(k),
\qquad
\mathcal C(k+2)=\mathcal C(k),
\end{equation}
and
\begin{equation}
\mathcal B(k+2)=\rho\,\mathcal B(k)
\end{equation}
for some constant $\rho$.  For distinct generic $\lambda_\mu$, these
conditions require respectively $\lambda_\mu^2=1$ for the active terms
of $\mathcal A$ and $\mathcal C$, and a common value of
$\lambda_\mu^2$ for the active terms of $\mathcal B$.  These are
non-generic restrictions.  Thus, although the mixed scalene
Yang--Baxter relation holds for arbitrary
$\lambda_\mu\in\C^\times$, the three constituent operators fail the
ordinary Yang--Baxter equation for generic deformation parameters.
\end{remark}

\subsection{Example : A four-generator realization}
\label{subsec:4geninfinite}
As a concrete illustration, consider the four-generator case $N=4$,
with $\lambda_1=1$, $\lambda_2=p$, $\lambda_3=r$, and $\lambda_4=s$,
where  $p,r,s\in\C^\times$. Then the generators of $\mathcal S_I$ and $\mathcal S_{III}$ are represented by the four bilateral weighted shifts $G_{1,j}$, $G_{p,j}$, $G_{r,j}$, and $G_{s,j}$, with respective weights $1, p^n, r^n, s^n$, and $n\in\mathbb Z$. Similar to this, the generators of $\mathcal S_{II}$ are represented by $H_{1,j}$, $H_{p,j}$, $H_{r,j}$, and $H_{s,j}$, with the associated weights
$1, p^{-n},r^{-n}, s^{-n}$. For generic pairwise distinct values of 1,$p$,$r$,$s$, all six pairwise exchange relations are represented nontrivially. For example,
\begin{equation*}
G_{p,j}G_{r,j}
=
\frac{p}{r}
G_{r,j}G_{p,j},
\qquad
G_{p,j}G_{s,j}
=
\frac{p}{s}
G_{s,j}G_{p,j},
\qquad
G_{r,j}G_{s,j}
=
\frac{r}{s}
G_{s,j}G_{r,j},
\label{eq:4infrelationsG}
\end{equation*}
whereas
\begin{equation*}
H_{p,j}H_{r,j}
=
\frac{r}{p}
H_{r,j}H_{p,j},
\qquad
H_{p,j}H_{s,j}
=
\frac{s}{p}
H_{s,j}H_{p,j},
\qquad
H_{r,j}H_{s,j}
=
\frac{s}{r}
H_{s,j}H_{r,j}.
\label{eq:4infrelationsH}
\end{equation*}

The corresponding scalene operators have a particularly simple action
on the tensor-product bases.  For $A_{12}$, define
\begin{equation}
\mathcal A(k)
=
a_1+a_2p^k+a_3r^k+a_4s^k.
\label{eq:Aweight}
\end{equation}
Using \eqref{eq:weightedA}, one obtains
\begin{equation}
A_{12}|m,n\rangle
=
\mathcal A(m-n)
|m+1,n+1\rangle.
\label{eq:Aaction}
\end{equation}
Here the two basis labels refer to sites $1$ and $2$.

For $B_{13}$, define
\begin{equation}
\mathcal B(k)
=
b_1+b_2p^k+b_3r^k+b_4s^k.
\label{eq:Bweight}
\end{equation}
Then
\begin{equation}
B_{13}|m,n\rangle
=
\mathcal B(m+n)
|m+1,n+1\rangle,
\label{eq:Baction}
\end{equation}
where the two basis labels refer to sites $1$ and $3$.

Similarly, define
\begin{equation}
\mathcal C(k)
=
c_1+c_2p^k+c_3r^k+c_4s^k.
\label{eq:Cweight}
\end{equation}
The action of $C_{23}$ is
\begin{equation}
C_{23}|m,n\rangle
=
\mathcal C(n-m)
|m+1,n+1\rangle,
\label{eq:Caction}
\end{equation}
where the two basis labels refer to sites $2$ and $3$.

Thus each of the three scalene operators is itself a bilateral weighted
shift on the corresponding two-site integer lattice, directed along the
diagonal translation
\begin{equation}
(m,n)\longmapsto(m+1,n+1).
\end{equation}
The scalar functions $\mathcal A(m-n)$, $\mathcal B(m+n)$ and $\mathcal C(n-m)$ play the role of the corresponding weight sequences.
Their algebraic invertibility reduces to the non-vanishing of these
weights. For example, $\mathcal A(k)\neq0$ for every $k\in\mathbb Z$,
is sufficient for $A_{12}$ to be algebraically invertible. Since
$m-n$ is unchanged by the simultaneous translation
$(m,n)\mapsto(m+1,n+1)$, its inverse is
\begin{equation}
A_{12}^{-1}|m,n\rangle
=
\frac{1}{\mathcal A(m-n)}
|m-1,n-1\rangle.
\label{eq:Ainverse}
\end{equation}
Analogous non-vanishing conditions apply to $B_{13}$ and $C_{23}$.

If $|p|=|r|=|s|=1$, then all the local generators, $G$'s and $H$'s,
are unitary on $\ell^2(\mathbb Z)$. Moreover, the sufficient condition $|a_1| > |a_2|+|a_3|+|a_4|$ implies $|\mathcal A(k)|\geq |a_1|-|a_2|-|a_3|-|a_4| >0$, for every $k\in\mathbb Z$. Hence $A_{12}$ possesses a bounded inverse on $\ell^2(\mathbb Z)\otimes\ell^2(\mathbb Z)$.  Analogous sufficient conditions hold for $B_{13}$ and $C_{23}$.

This four-generator example illustrates a principal advantage of the
infinite-dimensional representation.  The deformation parameters
$p,r,s$ may be completely generic, all local generators are
algebraically invertible on the bilateral domain, every pairwise
exchange relation is represented nontrivially, and unitary
representations are available whenever the parameters $\lambda_\mu$
lie on the unit circle.

Finally, a restricted obstruction can also be established for this
infinite-dimensional realization. Within the bilateral weighted-shift
ansatz
$R(u)|m,n\rangle=w_u(m,n)|m+1,n+1\rangle$,
the meromorphic difference-form YBE requires the matrix
$d_u(m,n)=w_u(m+1,n+1)/w_u(m,n)$ to have rank one,
where the relevant weights are nonzero. For the operator $B$,
this implies $B(4)B(1)^2=B(0)B(3)^2$, which fails generically
for the weights in \eqref{eq:Bweight}. Diagonal local similarities
preserve this rank condition. Thus the generic triple cannot
be reduced to a common spectral family within this ansatz
by diagonal local changes of basis. Reduction by arbitrary
non-diagonal local similarities and reduction to non-difference form $R$-matrices remain open questions.

\section{Conclusions}
\label{sec:conclusions}
In this work we have constructed a broad class of solutions of the
non-braided scalene Yang--Baxter equation from quadratic
noncommutative algebras. Beginning with two anticommuting generators,
we obtained a representation-independent construction which admits a
continuous deformation in terms of quantum-plane relations. The
general $N$-generator problem leads naturally to multiparameter
quantum affine spaces, with the scalene relation fixing their exchange
parameters to the ratio form
\begin{equation}
q^{\mu\nu}
=
\frac{\lambda_\mu}{\lambda_\nu},
\end{equation}
up to a common rescaling of the parameters $\lambda_\mu$.

A significant feature of the solutions obtained here is that they are
\emph{purely scalene} in the weak sense.  For, in all the explicit finite-dimensional
realizations considered in this paper, as well as in the
infinite-dimensional weighted-shift realization, the individual
operators $A$, $B$ and $C$ fail, for generic values of their parameters,
to satisfy the ordinary non-braided Yang--Baxter equation, while the
ordered triple $(A,B,C)$ satisfies the scalene Yang--Baxter relation.

Thus the mixed scalene relation holds even though, in the
chosen realizations, the individual constituents generically
fail the ordinary constant Yang--Baxter equation.%

Among all the solutions considered in this work we have also identified the purely scalene triples in the strong sense, where they fail to be recast as a spectral reparametrization of a Yang-Baxter family of solutions. These results are summarized in Table~\ref{tab:scalene-purity}.

\begin{table}[htbp]
\centering

\small
\renewcommand{\arraystretch}{1.25}
\begin{tabular}{@{}p{0.43\textwidth}p{0.19\textwidth}p{0.29\textwidth}@{}}
\hline
\textbf{Example}
& \textbf{Generically weakly pure}
& \textbf{Generically strongly pure} \\
\hline
Pauli realization, Example~2.4
& Yes & No \\

Singular qubit realization, $q^2\neq1$
& Yes & Yes \\

Three-generator qutrit realization
& Yes & Yes \\

Four-generator four-state realization
& Yes & Yes \\

Infinite-dimensional weighted shifts
& Yes & Unresolved; restricted obstruction established \\
\hline
\end{tabular}
\caption{Generic weak and strong purity of the explicit scalene
realizations, with strong purity understood in the sense defined
in the Introduction.}
\label{tab:scalene-purity}
\end{table}

We have also exhibited several concrete representations of the
underlying quantum affine algebras.  Finite-dimensional invertible
representations naturally lead to root-of-unity restrictions and can be
realized using finite Heisenberg--Weyl matrices.  By contrast, the
bilateral weighted-shift representation on $c_{00}(\mathbb Z)$ removes
the root-of-unity restriction and realizes arbitrary non-zero
deformation parameters.  Its equivalent realization on Laurent
polynomials connects the construction with the multiplicative shifts
underlying $q$-difference calculus.

The scalene relation also admits a simple mixed-$RLL$ interpretation.
Upon identifying
\begin{equation}
A_{ab}=R_{ab},
\qquad
B_{aj}=L^{(B)}_{aj},
\qquad
C_{bj}=L^{(C)}_{bj},
\end{equation}
it takes the form
\begin{equation}
R_{ab}L^{(B)}_{aj}L^{(C)}_{bj}
=
L^{(C)}_{bj}L^{(B)}_{aj}R_{ab}.
\end{equation}
Repeating this relation along a chain gives the corresponding relation
between the two monodromy operators.  Whenever the intertwiner is
invertible and the auxiliary traces are well defined, the associated
transfer operators satisfy the cross-commutativity relation
\begin{equation}
[t_B,t_C]=0.
\end{equation}
This is weaker than the usual self-commutativity condition of the
standard quantum inverse-scattering method, but it provides a natural
starting point for investigating the spectral and integrability
properties of scalene systems.

An important direction for future work is therefore to determine when
the purely scalene solutions constructed here can be used to generate
non-trivial integrable lattice models.  In particular, it would be
interesting to understand what additional algebraic or spectral
conditions, beyond cross-commutativity, are sufficient to produce
self-commuting transfer-matrix families and towers of independent
conserved quantities.  The examples obtained in the present work
provide explicit algebraic data with which these questions can be
studied.

\section*{Declaration of generative AI and AI-assisted technologies
in the manuscript preparation process}
The authors used ChatGPT (OpenAI) to assist with the organization
and exposition of the manuscript, the preparation of \LaTeX{} text,
bibliographic verification, and exploratory checks of mathematical
expressions and arguments. This assistance was guided by the authors'
research ideas, previous work, and scientific judgment. The authors
reviewed and revised the AI-assisted material and retain full
responsibility for the correctness, originality, interpretation,
and presentation of the results, as well as for the final content
of the manuscript.

\appendix
\section{Direct verification of the two-generator anticommuting solution}
\label{app:twogenerator}
For completeness, we give a direct verification of
Theorem~\ref{thm:anticomm}.  To simplify the notation within this
appendix, we introduce the shorthand
\begin{equation}
(x_I,y_I)
\equiv
\left(X_{1;I}^{(1)},X_{1;I}^{(2)}\right)~;~
(x_{II},y_{II})
\equiv
\left(X_{2;II}^{(1)},X_{2;II}^{(2)}\right)~;~
(x_{III},y_{III})
\equiv
\left(X_{3;III}^{(1)},X_{3;III}^{(2)}\right).
\end{equation}
The defining relations are
\begin{equation}
x_\alpha y_\alpha
=
-y_\alpha x_\alpha,
\qquad
\alpha\in\{I,II,III\},
\label{eq:appendixanti}
\end{equation}
while generators belonging to different lattice sites commute.

In this notation the three scalene operators are
\begin{equation}
A_{12}=a\,x_Ix_{II}+b\,y_Iy_{II}~;~
B_{13}=c\,x_Ix_{III}+d\,y_Iy_{III}~;~
C_{23}=e\,x_{II}x_{III}+f\,y_{II}y_{III}.
\end{equation}
Expanding the product $A_{12}B_{13}C_{23}$ and grouping together
operators acting on the same site gives
\begin{align}
A_{12}B_{13}C_{23}
={}&
ace\,
x_I^2x_{II}^2x_{III}^2
+
acf\,
x_I^2x_{II}y_{II}x_{III}y_{III}
\nonumber\\
&-
ade\,
x_Iy_Ix_{II}^2x_{III}y_{III}
+
adf\,
x_Iy_Ix_{II}y_{II}y_{III}^2
\nonumber\\
&+
bce\,
x_Iy_Ix_{II}y_{II}x_{III}^2
-
bcf\,
x_Iy_Iy_{II}^2x_{III}y_{III}
\nonumber\\
&+
bde\,
y_I^2x_{II}y_{II}x_{III}y_{III}
+
bdf\,
y_I^2y_{II}^2y_{III}^2.
\label{eq:antiexpansion}
\end{align}

Consider now the reverse product
$C_{23}B_{13}A_{12}$.  After grouping the factors site by site,
some of the local generator pairs occur in the reversed order.  For
example, the $acf$ contribution is initially
\begin{equation}
acf\,
x_I^2
\left(y_{II}x_{II}\right)
\left(y_{III}x_{III}\right).
\end{equation}
Using \eqref{eq:appendixanti} at sites $II$ and $III$ gives
\begin{equation}
\left(y_{II}x_{II}\right)
\left(y_{III}x_{III}\right)
=
\left(-x_{II}y_{II}\right)
\left(-x_{III}y_{III}\right),
\end{equation}
so that the two minus signs cancel and the contribution becomes
\begin{equation}
acf\,
x_I^2x_{II}y_{II}x_{III}y_{III},
\end{equation}
which agrees with the corresponding term in
\eqref{eq:antiexpansion}.

Similarly, the $adf$ contribution to the reverse product contains
\begin{equation}
adf\,
\left(y_Ix_I\right)
\left(y_{II}x_{II}\right)
y_{III}^2.
\end{equation}
Anticommuting the generators at sites $I$ and $II$ gives
\begin{equation}
\left(y_Ix_I\right)
\left(y_{II}x_{II}\right)
=
\left(-x_Iy_I\right)
\left(-x_{II}y_{II}\right),
\end{equation}
and hence
\begin{equation}
adf\,
x_Iy_Ix_{II}y_{II}y_{III}^2,
\end{equation}
again reproducing the corresponding term in
\eqref{eq:antiexpansion}.

The remaining coefficient sectors behave in the same way.  Whenever
the order of a pair of generators is reversed in passing from
$A_{12}B_{13}C_{23}$ to $C_{23}B_{13}A_{12}$, the reversal occurs at
an even number of sites.  The resulting anticommutation signs therefore
cancel pairwise.  Consequently,
\begin{equation}
C_{23}B_{13}A_{12}
=
A_{12}B_{13}C_{23},
\end{equation}
which proves the scalene Yang--Baxter relation.

\section{Connecting the representation in Sec. \ref{subsec:qdifference} to $q$-difference calculus}
\label{app:q-calculus}
We now explain the relation of this representation to $q$-difference
calculus.  In ordinary finite-difference calculus, one considers
additive translations of the argument,
\begin{equation}
f(x)\longmapsto f(x+a).
\end{equation}
In $q$-difference calculus the corresponding fundamental operation is
instead the multiplicative translation
\begin{equation}
f(x)\longmapsto f(qx).
\end{equation}
The operator implementing this transformation is precisely
$T_q$,
\begin{equation}
(T_qf)(x)=f(qx).
\end{equation}
For example, the standard $q$-difference derivative may be written as
\begin{equation}
(D_qf)(x)
=
\frac{f(qx)-f(x)}{(q-1)x},
\label{eq:qderivative}
\end{equation}
or, in operator notation,
\begin{equation}
D_q
=
\frac{1}{q-1}\,
M^{-1}(T_q-\mathbb I).
\label{eq:qderivativeoperator}
\end{equation}
In the limit $q\rightarrow1$, this reduces formally to the ordinary
derivative,
\begin{equation}
D_qf(x)
\longrightarrow
\frac{df(x)}{dx}.
\end{equation}
Thus the operators $G_\lambda$ and $H_\lambda$ are not themselves
$q$-difference derivatives; rather, they are constructed from the
multiplicative-shift operators that form the basic building blocks of
$q$-difference calculus.  In this sense the present construction gives
a $q$-difference realization of the quantum affine algebras.

\bibliographystyle{ieeetr}
\bibliography{refs}

@book{Baxter1982,
  author    = {Baxter, Rodney J.},
  title     = {Exactly Solved Models in Statistical Mechanics},
  publisher = {Academic Press},
  address   = {London},
  year      = {1982}
}

@book{KorepinBogoliubovIzergin1993,
  author    = {Korepin, Vladimir E. and Bogoliubov, Nikolay M. and Izergin, Anatoli G.},
  title     = {Quantum Inverse Scattering Method and Correlation Functions},
  publisher = {Cambridge University Press},
  address   = {Cambridge},
  year      = {1993},
  doi       = {10.1017/CBO9780511628832}
}

@incollection{Faddeev1996,
  author    = {Faddeev, Ludwig D.},
  title     = {How algebraic {Bethe} ansatz works for integrable models},
  booktitle = {Quantum Symmetries / Sym\'{e}tries quantiques},
  editor    = {Connes, Alain and Gawedzki, Krzysztof and Zinn-Justin, Jean},
  series    = {Les Houches Summer School Proceedings},
  volume    = {64},
  publisher = {North-Holland},
  address   = {Amsterdam},
  pages     = {149--219},
  year      = {1998},
  eprint    = {hep-th/9605187},
  archivePrefix = {arXiv},
  note = {arXiv:hep-th/9605187}
}

@article{HietarintaViallet2022,
  author  = {Hietarinta, Jarmo and Viallet, Claude},
  title   = {On the parametrization of solutions of the {Yang--Baxter} equations},
  journal = {Open Communications in Nonlinear Mathematical Physics},
  volume  = {2 154-172},
  year    = {2022},
  doi     = {10.46298/ocnmp.10204},
  eprint  = {q-alg/9504028},
  archivePrefix = {arXiv},
  note = {arXiv:q-alg/9504028}
}

@article{PadmanabhanKorepin2024Clifford,
  author  = {Padmanabhan, Pramod and Korepin, Vladimir},
  title   = {Solving the {Yang--Baxter}, tetrahedron and higher simplex equations using {Clifford} algebras},
  journal = {Nuclear Physics B},
  volume  = {1007},
  pages   = {116664},
  year    = {2024},
  doi     = {10.1016/j.nuclphysb.2024.116664},
  eprint  = {2404.11501},
  archivePrefix = {arXiv},
  primaryClass = {hep-th},
  note = {arXiv:2404.11501}
}

@article{PadmanabhanSinghKorepin2025CliffordSolver,
  author  = {Padmanabhan, Pramod and Singh, Vivek Kumar and Korepin, Vladimir E.},
  title   = {{Clifford} Solver for the Tetrahedron Equation and Its Variants},
  journal = {Bulgarian Journal of Physics},
  volume  = {52},
  number  = {s1},
  pages   = {126--131},
  year    = {2025},
  doi     = {10.55318/bgjp.2025.52.s1.126},
  eprint  = {2510.23944},
  archivePrefix = {arXiv},
  primaryClass = {hep-th},
  note    = {Proceedings of the XIII International Symposium on Quantum Theory and Symmetries (QTS-13), Yerevan, Armenia, 28 July--1 August 2025. arXiv:2510.23944}
}

@article{MaitySinghPadmanabhanKorepin2024,
  author  = {Maity, Somnath and Singh, Vivek Kumar and Padmanabhan, Pramod and Korepin, Vladimir},
  title   = {Algebraic classification of {Hietarinta}'s solutions of {Yang--Baxter} equations: invertible $4\times4$ operators},
  journal = {Journal of High Energy Physics},
  volume  = {2024},
  number  = {12},
  pages   = {067},
  year    = {2024},
  doi     = {10.1007/JHEP12(2024)067},
  eprint  = {2409.05375},
  archivePrefix = {arXiv},
  primaryClass = {hep-th},
  note = {arXiv:2409.05375}
}

@misc{PadmanabhanMaityKorepin2026,
  author  = {Padmanabhan, Pramod and Maity, Somnath and Korepin, Vladimir},
  title   = {Scalene {Yang--Baxter} triples as a source of hidden symmetries beyond the ordinary {Yang--Baxter} equation},
  year    = {2026},
  eprint  = {2608.09081},
  archivePrefix = {arXiv},
  note = {arXiv:2608.09081}
}

@incollection{FRT1989,
  author    = {Faddeev, L. D. and Reshetikhin, N. Yu. and Takhtajan, L. A.},
  title     = {Quantization of {Lie} groups and {Lie} algebras},
  booktitle = {Algebraic Analysis, Vol. I},
  editor    = {Kashiwara, M. and Kawai, T.},
  publisher = {Academic Press},
  address   = {Boston},
  pages     = {129--139},
  year      = {1988}
}

@article{Manin1987Koszul,
  author  = {Manin, Yuri I.},
  title   = {Some Remarks on {Koszul} Algebras and Quantum Groups},
  journal = {Annales de l'Institut Fourier},
  volume  = {37},
  number  = {4},
  pages   = {191--205},
  year    = {1987},
  doi     = {10.5802/aif.1117}
}

@book{Manin1988,
  author    = {Manin, Yuri I.},
  title     = {Quantum Groups and Noncommutative Geometry},
  series    = {Centre de Recherches Math\'ematiques Lecture Notes},
  publisher = {Universit\'e de Montr\'eal},
  address   = {Montr\'eal},
  year      = {1988}
}

@book{Manin2018QuantumGroups,
  author    = {Manin, Yuri I.},
  title     = {Quantum Groups and Noncommutative Geometry},
  series    = {CRM Short Courses},
  edition   = {2},
  publisher = {Springer},
  address   = {Cham},
  year      = {2018},
  doi       = {10.1007/978-3-319-97987-8},
  isbn      = {978-3-319-97987-8}
}

@article{Manin1989Multiparametric,
  author  = {Manin, Yuri I.},
  title   = {Multiparametric Quantum Deformation of the General Linear Supergroup},
  journal = {Communications in Mathematical Physics},
  volume  = {123},
  number  = {1},
  pages   = {163--175},
  year    = {1989},
  doi     = {10.1007/BF01244022}
}

@article{GoodearlLetzter2008,
  author  = {Goodearl, K. R. and Letzter, E. S.},
  title   = {Semiclassical limits of quantum affine spaces},
  journal = {Proceedings of the Edinburgh Mathematical Society},
  volume  = {52},
  number  = {2},
  pages   = {387--407},
  year    = {2009},
  eprint  = {0708.1091},
  archivePrefix = {arXiv},
  primaryClass = {math.QA},
  doi = {10.1017/S0013091507000910},
  note = {arXiv:0708.1091}
}

@article{MukherjeeBera2024,
  author  = {Mukherjee, Snehashis and Bera, Sanu},
  title   = {Construction of Simple Modules over the Quantum Affine Space},
  journal = {Algebra Colloquium},
  volume  = {31},
  number  = {1},
  pages   = {1--10},
  year    = {2024},
  doi     = {10.1142/S1005386724000026},
  eprint  = {2001.07432},
  archivePrefix = {arXiv},
  primaryClass = {math.RT},
  note = {arXiv:2001.07432}
}

@article{Schmudgen2002,
  author  = {Schm\"udgen, Konrad},
  title   = {On the Quantum Quarter Plane and the Real Quantum Plane},
  journal = {International Journal of Mathematics},
  volume  = {13},
  number  = {3},
  pages   = {279--321},
  year    = {2002},
  eprint  = {math/0005225},
  archivePrefix = {arXiv},
  note = {arXiv:math/0005225}
}

@misc{Hlavaty1997,
  author  = {Hlavat\'y, Ladislav},
  title   = {{Yang--Baxter} systems, solutions and applications},
  year    = {1997},
  eprint  = {q-alg/9711027},
  archivePrefix = {arXiv},
  note = {arXiv:q-alg/9711027}
}

@article{BrzezinskiNichita2005,
  author  = {Brzezi\'nski, Tomasz and Nichita, Florin F.},
  title   = {{Yang--Baxter} systems and entwining structures},
  journal = {Communications in Algebra},
  volume  = {33},
  number  = {4},
  pages   = {1083--1093},
  year    = {2005},
  eprint  = {math/0311171},
  archivePrefix = {arXiv},
  primaryClass = {math.QA},
  note = {arXiv:math/0311171}
}

@article{BerceanuNichitaPopescu2013,
  author  = {Berceanu, Barbu R. and Nichita, Florin F. and Popescu, C\u{a}lin},
  title   = {Algebra Structures Arising from {Yang--Baxter} Systems},
  journal = {Communications in Algebra},
  year    = {2013},
  doi     = {10.1080/00927872.2012.703736},
  eprint  = {1005.0989},
  archivePrefix = {arXiv},
  primaryClass = {math.QA},
  volume  = {41},
  number  = {12},
  pages   = {4442--4452},
  note = {arXiv:1005.0989}
}

@misc{Isaev2022,
  author  = {Isaev, A. P.},
  title   = {Lectures on quantum groups and {Yang--Baxter} equations},
  year    = {2022},
  eprint  = {2206.08902},
  archivePrefix = {arXiv},
  primaryClass = {math.QA},
  note = {arXiv:2206.08902}
}

@article{OhnukiKitakado1993,
  author  = {Ohnuki, Y. and Kitakado, S.},
  title   = {Fundamental Algebra for Quantum Mechanics on {$S^D$} and Gauge Potentials},
  journal = {Journal of Mathematical Physics},
  volume  = {34},
  number  = {7},
  pages   = {2827--2851},
  year    = {1993},
  doi     = {10.1063/1.530099}
}

@article{CarruthersNieto1968,
  author  = {Carruthers, P. and Nieto, Michael Martin},
  title   = {Phase and Angle Variables in Quantum Mechanics},
  journal = {Reviews of Modern Physics},
  volume  = {40},
  number  = {2},
  pages   = {411--440},
  year    = {1968},
  doi     = {10.1103/RevModPhys.40.411}
}

@book{Vaidyanathan2023FunctionalAnalysis,
  author    = {Vaidyanathan, Prahlad},
  title     = {Functional Analysis},
  publisher = {Cambridge University Press},
  year      = {2023},
  isbn      = {9781009243902},
  doi       = {10.1017/9781009243926}
}

@article{Bergman1978Diamond,
  author  = {Bergman, George M.},
  title   = {The {Diamond} Lemma for Ring Theory},
  journal = {Advances in Mathematics},
  volume  = {29},
  number  = {2},
  pages   = {178--218},
  year    = {1978},
  doi     = {10.1016/0001-8708(78)90010-5}
}

@article{Oh2008QuantumPoisson,
  author  = {Oh, Sei-Qwon},
  title   = {Quantum and {Poisson} Structures of Multi-Parameter
             Symplectic and {Euclidean} Spaces},
  journal = {Journal of Algebra},
  volume  = {319},
  number  = {11},
  pages   = {4485--4535},
  year    = {2008},
  doi = {10.1016/j.jalgebra.2008.03.010}
}

@misc{konstantinou2026scalene,
  author = {Konstantinou-Rizos, S. and Kouloukas, T.},
  title = {Scalene {Yang--Baxter} maps and {Lax} triples},
  year = {2026},
  eprint = {2604.27060},
  archivePrefix = {arXiv},
  primaryClass = {nlin.SI},
  note = {arXiv:2604.27060}
}

@article{BaxterPerkAuYang1988,
  author  = {Baxter, R. J. and Perk, J. H. H. and Au-Yang, H.},
  title   = {New solutions of the star-triangle relations
             for the chiral {Potts} model},
  journal = {Physics Letters A},
  volume  = {128},
  number  = {3--4},
  pages   = {138--142},
  year    = {1988},
  doi     = {10.1016/0375-9601(88)90896-1}
}

@incollection{PerkAuYang2006,
  author        = {Perk, J. H. H. and Au-Yang, H.},
  title         = {{Yang--Baxter} Equations},
  booktitle     = {Encyclopedia of Mathematical Physics},
  editor        = {Fran{\c{c}}oise, J.-P. and Naber, G. L.
                   and Tsou, S. T.},
  publisher     = {Elsevier},
  address       = {Oxford},
  volume        = {5},
  pages         = {465--473},
  year          = {2006},
  doi           = {10.1016/B0-12-512666-2/00191-7},
  eprint        = {math-ph/0606053},
  archivePrefix = {arXiv}
}

\end{document}